\documentclass[11pt]{article}

\usepackage{cite}
\usepackage[hidelinks]{hyperref}
\usepackage{booktabs}
\usepackage{paralist}
\usepackage{algorithm}
\usepackage[noend]{algpseudocode}
\usepackage{apptools}
\usepackage{fancyvrb}
\usepackage{tikz}
\usepackage{multirow}
\usepackage{graphicx}
\usepackage{mathtools}
\usepackage{amsthm}
\usepackage{adjustbox}
\usepackage{amssymb}
\usepackage{relsize}
\usepackage{titlesec}
\usepackage[caption=false,font=normalsize,labelfont=sf,textfont=sf]{subfig}

\usepackage{float}
\newfloat{algorithm}{t}{lop}

\titlespacing*{\paragraph}
{0pt}   
{0.5em}   
{0.5em} 

\newcommand{\ifstringequal}[4]{%
  \ifnum\pdfstrcmp{#1}{#2}=0
  #3%
  \else
  #4%
  \fi
}

\algnewcommand\True{\textbf{true}}
\algnewcommand\False{\textbf{false}}
\algnewcommand\Input{\item[\textbf{Input:}]}
\algnewcommand\Output{\item[\textbf{Output:}]}

\newtheorem{definition}{Definition}
\newtheorem{proposition}{Proposition}
\newtheorem{lemma}{Lemma}
\newtheorem{theorem}{Theorem}
\newtheorem{example}{Example}
\newtheorem{proofpart}{Part}

\renewcommand{\epsilon}{\varepsilon}
\renewcommand{\phi}{\varphi}
\DeclareMathOperator*{\argmax}{arg\ max}

\newcommand{\name}[1]{\mathrm{#1}}

\newcommand{\smartparen}[3]{%
  \mathchoice%
  {\left#1#3\right#2}
  {#1#3#2}
  {#1#3#2}
  {#1#3#2}
}

\newcommand{\parenth}[1]{\smartparen{(}{)}{#1}}
\newcommand{\abs}[1]{\smartparen{|}{|}{#1}}
\newcommand{\ceil}[1]{\left\lceil{#1}\right\rceil}
\newcommand{\anglebr}[1]{\langle{#1}\rangle}
\newcommand{\curly}[1]{\left\{#1\right\}}
\newcommand{\set}[2]{\curly{#1 \ \middle|\ #2}}
\newcommand{\range}[2]{\curly{#1, \dotsc, #2}}
\newcommand{\rowvector}[2]{(#1_1, \dotsc, #1_#2)}
\newcommand{\conj}{\wedge}
\newcommand{\disj}{\vee}
\newcommand{\bigconji}[3]{\bigwedge\nolimits_{#1}^{#2} #3}
\newcommand{\bigdisji}[3]{\bigvee\nolimits_{#1}^{#2} #3}
\newcommand{\kdnfi}{\bigdisji{i=1}{m}{\parenth{\bigconji{j=1}{k}{a_j ^ i}}}}
\newcommand{\true}{1}
\newcommand{\false}{0}

\makeatletter
\newcommand{\casesfunc}[2]{\begin{cases} #2 & #1 \casesfuncchecknextarg}
\newcommand{\casesfuncchecknextarg}{\@ifnextchar\bgroup{\casesfunchandlenextarg}{\end{cases}}}
\newcommand{\casesfunchandlenextarg}[2]{\\ #2 & #1 \@ifnextchar\bgroup{\casesfunchandlenextarg}{\end{cases}}}
\makeatother

\newcommand{\func}[2]{\name{#1}\parenth{#2}}
\newcommand{\To}{\rightarrow}
\newcommand{\cardinality}[1]{\abs{#1}}

\renewcommand{\epsilon}{\varepsilon}
\renewcommand{\phi}{\varphi}

\mathchardef\mhyphen="2D

\newcommand{\thetaof}[1]{\(\mathrm{\Theta}\!\left({#1}\right)\)}

\newcommand{\vars}[1]{\func{Vars}{#1}}

\newcommand\restr[2]{{
  \left.\kern-\nulldelimiterspace
  #1
  \vphantom{\big|}
  \right|_{#2}
}}



\newcommand{\des}{DES}
\newcommand{\fscore}{\(F_1\)}

\begin{document}
\newcommand{\titleFirstLine}{Query-Guided Analysis and Mitigation of}
\newcommand{\titleSecondLine}{Data Verification Errors (Extended Version)}
\title{\titleFirstLine{} \\ \titleSecondLine{} \thanks{This paper is an extended version of our paper published in the IEEE International Conference on Data Engineering (ICDE) 2026. DOI: \href{https://doi.org/10.1109/ICDE65706.2026.00151}{10.1109/ICDE65706.2026.00151}.}}

\author{Ran Schreiber \\
The Department of Computer Science \\
Bar-Ilan University, Ramat Gan, Israel\\
ransch7@gmail.com
\and
Yael Amsterdamer \\
The Department of Computer Science \\
Bar-Ilan University, Ramat Gan, Israel\\
yael.amsterdamer@biu.ac.il
}

\maketitle

\begin{abstract}
Data verification, the process of labeling data items as correct or incorrect, is a preprocessing step that may critically affect the quality of results in data-driven pipelines. Despite recent advances, verification can still produce erroneous labels that propagate to downstream query results in complex ways. We present a framework that complements existing verification tools by assessing the impact of potential labeling errors on query outputs and guiding additional verification steps to improve result reliability. To this end, we introduce \emph{Maximal Error Score (MES)}, a worst-case uncertainty metric that quantifies the reliability of query output tuples independently of the underlying data distribution. As an auxiliary indicator, we identify \emph{risky tuples} -- input tuples for which reducing label uncertainty may counterintuitively increase the output uncertainty. We then develop efficient algorithms for computing MES and detecting risky tuples, as well as a generic algorithm, named MESReduce, that builds on both indicators and interacts with external verifiers to select effective additional verification steps. We implement our techniques in a prototype system and evaluate them on real and synthetic datasets, demonstrating that MESReduce can substantially and effectively reduce the MES and improve the accuracy of verification results.

\end{abstract}

\section{Introduction}
\label{sec:introduction}

\emph{Data verification} is the process of classifying data items as correct or incorrect, which can be critical in data-driven applications that require high data accuracy.  Verifying the data may involve fact-checking against external knowledge sources or detecting anomalies in format or content, and thus can be viewed as part of the broader data cleaning process. Recent advances in machine learning have led to growing interest and improved performance in this area (see Section~\ref{sec:related_work}).

Despite these advances, data verification can be costly in terms of latency and resources, e.g., when each tuple must be fact-checked by searching a data lake, using a Large Language Model, or consulting a human expert. Moreover, the resulting labels may still be uncertain or erroneous, especially under cost-accuracy tradeoffs. For example, data lake searches may be time-limited; verification systems that rely on human feedback, e.g., by involving domain experts or crowd workers~\cite{yakout2011guided, drien2023query, bergman2015query, krishnan2016activeclean, assadi2018cleaning,cheng2008cleaning, mo2013cleaning, wang2014sample, fariha2021coco, chu2015katara, breck2019data, schelter2018automating}, are naturally prone to erroneous decisions -- correct data might be labeled as incorrect or vice versa; similarly, systems in which labels are produced automatically, e.g., by an ML classifier, are also prone to erroneous labels. Such verification errors can then propagate, often in non-trivial ways, to downstream tasks~\cite{khan2025still,neutatz2021cleaning,nguyen2019user}. In the context of queries issued by analysts (which we show to be affected non-trivially by the verification errors in Section~\ref{sec:mes_metric_risky_tuples}), a notable shortcoming of existing systems is that the potential error propagation is not quantified and therefore not reported properly to the analysts, which can be particularly problematic in data-critical settings, e.g., medical applications.

\begin{example}
 \label{ex:analyst}
Consider a data analyst deriving business insights from a database that contains both valuable data and errors. To focus on the correct portion of the data, the analyst applies an automated data verification tool that labels each tuple as correct or incorrect. However, because verification is imperfect, these labels may themselves be erroneous and can propagate to downstream query results. This raises several key questions: \begin{inparaenum}[(i)] \item To what extent can the analyst trust the query results? \item How likely is it that correct output tuples are missing from the results? \item What actions can be taken to improve the reliability of the results?\end{inparaenum}

Answering these questions is not straightforward. Consider questions (i)-(ii) and suppose the query aims to retrieve data about company founders. For some individual~$X$, the database contains ten distinct tuples of the form ``$X$ founded company~$Y_i$,'' where $Y_1, \ldots, Y_{10}$ are different companies. The higher-level conclusion that~$X$ is a company founder is supported if at least one of these tuples is correct. If all ten tuples are labeled as correct, then all ten labels would have to be erroneous to invalidate this conclusion. In contrast, if all ten tuples are labeled as incorrect, a single mislabeled tuple would be sufficient to invalidate the opposite conclusion that~$X$ is not a company founder.

Now consider question~(iii) and assume that the analyst can hire a domain expert to refine the labels. Verifying the entire dataset is costly, so the analyst must decide which tuples the expert should check. Prioritizing tuples solely based on their individual error likelihood may be ineffective: e.g., if the tuple ``$X$ founded company $Y_1$'' is labeled as correct with (nearly)~100\% certainty, then~$X$ is almost certainly a company founder, and refining the remaining nine labels is redundant regardless of their uncertainty.\hfill\qed{}
\end{example}

In this work, we present a generic framework that complements existing data verification tools by analyzing how uncertainty in the labels they produce affects downstream query results. First, we introduce a novel uncertainty measure, \emph{Maximal Error Score (MES)}, which assigns each output tuple a worst-case error score based on the uncertainty of the labels assigned to input tuples. This measure highlights where verification provides strong reliability guarantees for query outputs and where outputs remain uncertain. Next, in settings where label error probabilities can be reduced at an additional cost, such as in Example~\ref{ex:analyst}, we propose a generic algorithm that leverages our uncertainty indicators to automatically guide external verifiers and optimize the tradeoff between query result reliability and verification cost.

Example uses of the framework include estimating the reliability of query output, considering both tuples labeled as correct and as incorrect; determining whether a fast, low-cost data verification tool produces sufficiently reliable results or if a more accurate, costlier tool is needed; and selecting which tuples should be labeled by a high-quality verifier under a limited budget.

\begin{figure}[t]
\center
\includegraphics[width=0.55\linewidth]{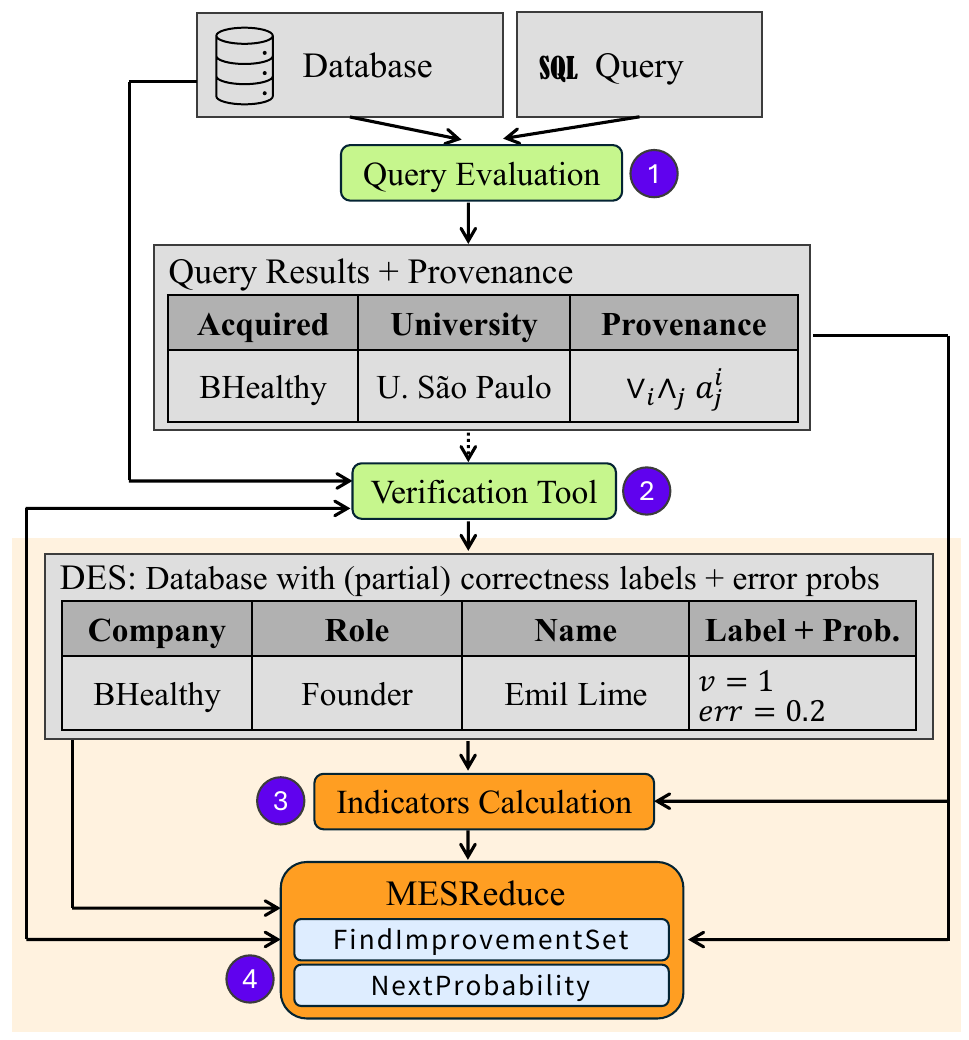}
\caption{The flow of our framework}
\label{fig:architecture_diagram}
\end{figure}

\subsection{Framework Components and Flow}
Figure~\ref{fig:architecture_diagram} illustrates the high-level flow of our framework. Novel components are enclosed in a highlighted background, and step numbers are shown in circles.

\textbf{Step~1}: A user-specified query is evaluated together with its \emph{provenance information}, which abstractly captures how output tuples are derived and is used in subsequent steps.

\textbf{Step~2}: An external verification tool is applied to produce a (possibly partial) correctness labeling of input tuples, along with an estimated probability for each label (shown in the Label+Prob.\ column in the figure). Section~\ref{sec:model} details how these probabilities can be estimated. The resulting instance is referred to as a \emph{database with error estimation (DES)}. If the verifier is \emph{query-guided}, it also receives the query results as input. Output tuples are presented to the analyst together with derived correctness labels.

\textbf{Step~3}: The framework provides two novel indicators to support analyst decision making. Based on these, the analyst may retain only reliable, low-MES (i.e., low-uncertainty) output tuples, apply a different verifier, or use the \emph{risky tuple} indicator, which identifies tuples whose further verification may increase overall uncertainty, to prioritize non-risky tuples for refinement.

\textbf{Step~4}: Optionally, the analyst can invoke our \emph{MESReduce} algorithm, which optimizes the verification quality-cost tradeoff by iteratively selecting input tuples for refinement, applying external verification tools, and gradually reducing uncertainty in the query results.

\subsection{Uncertainty Indicators}
We now outline the key ideas behind our uncertainty indicators and how they differ from prior work, detailed in Section~\ref{sec:indicators}. Our output uncertainty metric, MES, adopts a \emph{worst-case perspective} computing an upper bound on the probability that an output tuple is mislabeled across all possible worlds. This score provides a \emph{strong reliability guarantee}: tuples with low MES are unlikely to be mislabeled regardless of the ground truth.
Unlike prior work on uncertain databases, which estimates expected uncertainty and therefore requires a probability distribution over the ground truth (which is sometimes unavailable or hard to obtain), MES relies solely on the inherent uncertainty of the data verification process.

We observe that reducing the error likelihood of input tuple labels does not always decrease output uncertainty, as reflected in its MES; in fact, it may even increase it. We call such input tuples \emph{risky}, and identifying them is crucial for planning additional verification efforts.

\subsection{Contributions Overview and Novelty}
\label{sec:contributions}
Our work aims at assessing and improving the reliability of query results under uncertain data verification. Unlike prior approaches that ignore verification errors or implicitly try to minimize them, we \emph{explicitly} model how such errors propagate through queries and provide quantitative guarantees on result reliability. Our main contributions are as follows.

\paragraph*{Formal model of uncertain verification}
We present a formal model for uncertainty in data verification results. In contrast with prior work on uncertain, probabilistic, and possibilistic databases, our approach builds directly on uncertainty in the verification process rather than on an estimation of the ground truth data.

\paragraph*{Novel reliability indicators}
We propose two new metrics tailored to this setting: \emph{Maximal Error Score (MES)} and \emph{risky tuples}, which help users assess query-aware reliability and support multiple practical uses. Fundamentally, their formulation clarifies how verification errors propagate to query results.

\paragraph*{Efficient indicator computation}
We analyze the complexity of MES computation and provide efficient algorithms for computing our novel indicators.

\paragraph*{Uncertainty reduction algorithm}
Based on the concepts of MES and risky tuples, we develop a generic uncertainty reduction algorithm called \emph{MESReduce} that builds on top of verification systems of the analyst's choice.

\paragraph*{Prototype implementation and evaluation}
We implemented our techniques in a prototype system and evaluated them on both real and synthetic datasets. Section~\ref{sec:case_studies} presents multiple end-to-end case studies that illustrate how our framework integrates with prominent verification approaches and improves their effectiveness in practice. In the first case study, mirroring Example~\ref{ex:analyst}, an analyst initially applies a low-cost automated verifier to the entire dataset. Our framework then identifies which labels should be refined to maximize query result reliability. We demonstrate two refinement strategies: deciding how many crowd workers to allocate to selected tuples, and selecting tuples to be re-verified using a high-quality LLM. In both settings, our approach improves over initial verification results and achieves high output reliability at a substantially lower cost than exhaustive verification. Another case study shows how a lower-cost LLM can be used for initial verification with probability estimation, and how our framework guides targeted refinement via crowdsourcing to further improve reliability.

Beyond these concrete use cases, we conduct an extensive experimental study to evaluate the robustness, scalability, and effectiveness of our framework across diverse scenarios. Because our goal is not to replace existing verification systems but to complement and guide their use, we compare our approach against alternative strategies for interacting with such tools rather than against the tools themselves. The results show that our methods consistently improve verification reliability and output quality under limited budgets, scale to large query results and complex provenance structures, and offer practical guidance for configuring the framework to effectively interact with real verification systems.

Overall, our work provides the first formal framework for reasoning about and improving query reliability under uncertain data verification. It provides both theoretical contributions and practical applicability; we anticipate that the contributions will support analysts in improving the effectiveness and reliability of their analytical workflows under uncertainty. A short paper describing a demonstration of our framework, along with a prototype GUI system for analysts, was published in~\cite{schreiber2024cleaner}.

\section{Model}
\label{sec:model}

\begin{table}[t]
\caption{Notation Summary}
\label{tab:notation}
\centering
\begin{tabular}{lp{0.5\linewidth}}
\hline
\textbf{Symbol} & \textbf{Description} \\
\hline
\(Q\) & An SPJU query \\
\(D\) & A relational database \\
\(\overline{D}\) & A \des{} (Def.~\ref{def:des}) \\
\(\hat{D}\) & Annotated \des{} (Def.~\ref{def:annotated_des}) \\
\(v(t)\) & A correctness label given to an input tuple~\(t\) \\
\(err(t)\) & An error probability of a label given to~\(t\) \\
\(v_\name{full}\) & A possible world (Def.~\ref{def:possible_worlds}) \\
\(v^Q(o)\) & The label of an output tuple~\(o\) \\
\(L(t)\) & The variable annotating an input tuple~\(t\) \\
\(L^Q (o)\) & The provenance exp. of an output tuple~\(o\) (Def.~\ref{def:annotated_des_query}) \\
\(P\parenth{\overline{D}, v_\name{full}, D'}\) & Labeling probability (Def.~\ref{def:labeling_prob}) \\
\hline
\end{tabular}
\end{table}

\begin{table*}[t]
\caption{A Relational database~$D_{\name{ex}}$ -- the input to the data verification process, here including tuple annotations.}
\vspace{3mm}
  \small
  \resizebox{\linewidth}{!}{\begin{tabular}{ccc|c}
    \toprule
    \multicolumn{4}{c}{\textbf{Acquisitions}} \\
    \midrule
    Acquired & Acquiring & Date & \(L\) \\
    \midrule
    BHealthy & Fiffer & 04.03.2018 & \(a_1\) \\
    NewHealth  & BHealthy & 02.04.2017 & \(a_2\) \\
    NewHealth & Optobest & 01.10.2020 & \(a_3\) \\
    Optobest & microBarg & 08.08.2019 & \(a_4\) \\
    \bottomrule
  \end{tabular}
  \quad
  \begin{tabular}{ccc|c}
    \toprule
    \multicolumn{4}{c}{\textbf{Roles}} \\
    \midrule
    Company & Role & Name & \(L\) \\
    \midrule
    BHealthy & Founder & Emil Lime & \(r_1\) \\
    NewHealth & Co-founder & Yara Aray & \(r_2\) \\
    Optobest & Co-founder & Nala Alan & \(r_3\) \\
    BHealthy & Co-founder & Rima Amir & \(r_4\) \\
    \bottomrule
  \end{tabular}
  \quad
  \begin{tabular}{ccc|c}
    \toprule
    \multicolumn{4}{c}{\textbf{Education}} \\
    \midrule
    Name & University & Year & \(L\) \\
    \midrule
    Yara Aray & U.\ Melbourne & 2017 & \(e_1\) \\
    Emil Lime & U.\ São Paulo & 2014 & \(e_2\) \\
    Rima Amir & U.\ São Paulo  & 2018 & \(e_3\) \\
    Nala Alan & U.\ Cape Town & 2002 & \(e_4\) \\
    \bottomrule
  \end{tabular}}
  \label{tab:des_example}
\end{table*}

\begin{table}[t]
  \caption{Uncertain labels produced by a verification tool for~$D_{\name{ex}}$. {\small {\normalfont The variables on the header row are the tuple annotations from Table~\ref{tab:des_example}. The first row includes correctness labels for some of the tuples ($\bot$~stands for unassigned labels); the second includes error probabilities for each label; and the third demonstrates a possible ground truth.}}}
  \centering
  \begin{tabular}{ccccccccccccc}
    \toprule
  & $a_1$ & $a_2$ & $a_3$ & $a_4$ & $r_1$ & $r_2$ & $r_3$ & $r_4$ & $e_1$ & $e_2$ & $e_3$ & $e_4$ \\
    \midrule
  $v_{\name{ex}}:$ & \true & \false & \true & \false &\true & $\bot$ & $\bot$ & $\bot$ & $\bot$ & \true & $\bot$ & $\bot$ \\
  $\name{err}$: & 0.3 & 0.3 & 0.2 & 0 & 0.2 & $\bot$ & $\bot$ & $\bot$ & $\bot$ & 0.4 & $\bot$ & $\bot$ \\
  $v_\name{full}$: & \true & \false & \false & \false &\false & \false & \false & \false & \false & \true & \false & \false \\
    \bottomrule
  \end{tabular}
  \label{tab:examplenumbers}
\end{table}

In this section, we introduce a formal model for data verification, under which we will later analyze potential verification errors. We start by providing a formal generic definition of the possibly erroneous output of data verification. 

\paragraph*{Database with error estimation (\des{})}
We view data verification as a probabilistic process in which correctness labels are given to tuples, and each label has some likelihood of an error (Step~2 in Fig.~\ref{fig:architecture_diagram}). Therefore, the output of such a process is called a \emph{database with error estimation (\des{})} and is defined as follows.
\begin{definition}
\label{def:des}
A \des{} is a triplet \(\overline{D} = {\anglebr{D, v, \name{err}}}\) where $D$ is a relational database, \(v(t)\) is the correctness label given to a tuple~\(t \in D\), and \(err(t)\) is the error probability of that label. Formally, \(v: D\to \curly{\false{}, \true{}, \bot}\) is a 3-valued labeling function representing the correctness labels assigned by the verification process to~\(D\)'s tuples (where~\false{}, \true{} and~$\bot$ stand for correct, incorrect and unknown, i.e., not labeled, resp.), and \(\name{err}: D \to [0,0.5] \cup \curly{\bot}\) assigns an error probability\footnote{We assume w.l.o.g.\ that the error probability of a label is up to~\(0.5\); otherwise, one may negate the label to achieve this.} to each tuple~\(t \in D\), such that \begin{inparaenum}[(i)] \item if $v(t)\in\{\false,\true\}$ then $\func{err}{t}$ is the probability that $v(t)=\neg v^*(t)$ for the (unknown) ground truth labeling~$v^*: D\to \curly{\false{}, \true{}}$; and \item $\func{err}{t}=\bot$ iff $v(t)=\bot$.\end{inparaenum}
\end{definition}

Intuitively, the output of data verification is a partial correctness labeling, where each label may be erroneous and is associated with an estimated error probability. This definition is intentionally broad and makes no assumptions about the underlying data or verification method, making it applicable across a wide range of systems (see Section~\ref{sec:related_work}).

Crucially, these probabilities apply to errors in the \emph{verification process} itself, rather than in the input data being verified. This distinction is illustrated in the following example.

\begin{example}[DES]
\label{example:des}
Table~\ref{tab:des_example} presents a database $D_{\name{ex}}$ with three relations:  \textbf{Acquisitions}, with information about company acquisitions; \textbf{Roles}, with information about roles in companies; and \textbf{Education}, with information about university alumni. Column~\(L\) includes a unique variable annotating each tuple (to be defined in Def.~\ref{def:annotated_des}), and we will use these annotations to refer to specific tuples. This database is fed as input to a data verification system, resulting in a DES $\overline{D}_{\name{ex}}=(D_{\name{ex}},v_{\name{ex}},\name{err})$ with $v_{\name{ex}}$ and $\name{err}$ specified in Table~\ref{tab:examplenumbers} (ignore, for now, the last row). For instance, we have $v_{\name{ex}}(r_1)=1$ and $\name{err}(r_1)=0.2$, meaning that the tuple annotated by~$r_1$ is estimated to be correct by the verification system and that with probability~$0.2$ this label is wrong, i.e., $r_1$~is incorrect. Note that this does not necessarily imply that~$r_1$ is more likely to be correct than incorrect: for example, if~$r_1$ was produced by a very noisy source that has a probability~0.1 of being correct, then we would observe a true ``correct'' label with probability~$0.1\cdot 0.8=0.08$, whereas having an erroneous ``correct'' label would have a higher probability of $0.9\cdot0.2=0.18$.

As another example, $v_{\name{ex}}(a_2)=0$ and~$\name{err}(a_2)=0.3$, meaning that the corresponding tuple is estimated to be incorrect, with a likelihood of~0.3 of having an erroneous label; $v_{\name{ex}}(e_4)=\bot$ means that this tuple was not labeled by the verification system and we have no information regarding its correctness.
\end{example}

\paragraph*{Cell-level uncertainty}
Throughout the paper, we assume tuple-level uncertainty. A practical way of applying our work to cell-level uncertainty is to represent each tuple as triplets of Tuple-ID, Attribute, and Value (assuming a primary key exists). In the new database, tuple-level labels directly reflect cell correctness in the original database. In addition, every SPJU query over the original database can then be reconstructed through joins to operate over the new one.

\begin{example}
The Acquisitions relation of Table~\ref{tab:des_example} would be transformed to a relation AcquisitionsC with three attributes, ID, Attribute, and Value. In particular, the first tuple would be transformed to three tuples: (1, Acquired, BHealthy), (1, Acquiring, Fiffer), (1, Date, 04.03.2018).

Now, consider the following query over the original database:
\begin{smaller}
\begin{Verbatim}[frame=single]
SELECT DISTINCT r.Name, a.Acquiring
FROM Acquisitions AS a, Roles AS r
WHERE a.Acquired = r.Company
\end{Verbatim}
\end{smaller}

It will be transformed to the following query:
\begin{smaller}
\begin{Verbatim}[frame=single]
SELECT DISTINCT r2.Value, a2.Value
FROM AcquisitionsC as a1,
     AcquisitionsC as a2,
     RolesC as r1,
     RolesC as r2,
WHERE a1.ID = a2.ID AND
      r1.ID = r2.ID AND
      a1.Attribute = 'Acquired' AND
      a2.Attribute = 'Acquiring' AND
      r1.Attribute = 'Company' AND
      r2.Attribute = 'Name' AND
      a1.Value = r1.Value
\end{Verbatim}
\end{smaller}
\end{example}

\paragraph*{Obtaining error probabilities}
We now outline practical methods for obtaining error probabilities based on the verifiers. Some systems natively provide error estimates (e.g.,~\cite{krishnan2016activeclean,wang2014sample}). For verifications involving human experts or crowdsourcing, various established models can be used to estimate worker or label accuracy (e.g.,~\cite{dawid1979maximum,kerrigan2021combining,li2015survey}). For anomaly detection models, the predicted anomaly scores can be translated into estimated error probabilities, as demonstrated in Section~\ref{sec:case_studies}. Similarly, machine learning classifiers often output probabilities for each class, and the one assigned to the inverse of the predicted label can serve as an error estimate. For LLMs, there are various techniques of estimation, and we demonstrate the use of one such approach in Section~\ref{sec:case_studies} (e.g.,\cite{kadavath2022language, xiong2024can, kohn2023semantic, xiong2024efficient, lin2022teaching, shorinwa2026survey}). Lastly, when relying on external data sources, published precision metrics may be used as proxies for error probabilities (e.g., as in~\cite{mitchell2015nell}).

Note that even if we have no knowledge of the uncertainty of the initial verification results, our uncertainty reduction algorithm in Section~\ref{sec:uncertainty_red_algs} can still be valuable; see Scenario~\texttt{WCS} in Section~\ref{sec:experiments_setting}.

\paragraph*{Possible world semantics}
To capture the possible ground truth for a given DES, we use possible world semantics, following some of the previous work on uncertain databases~\cite{suciu2011prob,drien2023query,greco2014certain,lian2010consistent,suciu2011prob}. A possible world is a subset of the input tuples that are correct in this world; the set of all possible worlds is the set of every such subset. Equivalently, it is a full labeling of the data (without~$\bot$), where a tuple is assigned as \(\true\) iff it is correct in the world. Formally,

\begin{definition}[Possible worlds and label errors]
\label{def:possible_worlds}
Let \(\overline{D} = {\anglebr{D, v, \name{err}}}\) be a \des{}, a \emph {possible world for~\(\overline{D}\)} is a Boolean labeling function \(v_\name{full}:D\To\{\false,\true\}\) where a tuple \(t \in D\) is \emph{correct in \(v_\name{full}\)} iff \(v_\name{full}(t)=\true\). A label \(v(t)\) of the \des{} is then \emph{erroneous in \(v_\name{full}\)} iff $v(t)=\neg v_\name{full}(t)$.
\end{definition}

\begin{example}[Possible world]
\label{example:possible_worlds}
Continuing Example~\ref{example:des}, consider the possible world \(v_\name{full}\) depicted in Table~\ref{tab:examplenumbers}. With respect to this world, the label of~$r_1$ is erroneous (the tuple is labeled by the DES as correct but is actually incorrect), and for tuple~$a_2$ (that is labeled by the DES as incorrect) and~$r_2$ (that is not labeled), the labels are not erroneous.  
\end{example}

We can use the error probabilities in the \des{}, to compute the probability of a given labeling in a possible world:

\begin{definition}[Labeling probability]
\label{def:labeling_prob}
Let \(\overline{D}={\anglebr{D, v, \name{err}}}\) be a \des{}, let \(v_\name{full}\) be a possible world for \(\overline{D}\) and let \(D'\subseteq D\) be a nonempty subset of the database. Assuming that error probabilities in $\name{err}$ are independent, one can compute the probability of obtaining $v$ from~$v_\name{full}$ given~$\name{err}$ and restricting attention to $D'$ as
\begin{equation}
\label{equation:labeling_prob}
P\parenth{\overline{D}, v_\name{full}, D'} = \prod_{t \in T_{\name{right}}} (1 - \func{err}{t}) \cdot \prod_{t \in T_{\name{err}}} \func{err}{t}
\end{equation}
where we use, for short, \(T_{\name{right}} = \set{t \in D'}{v(t)=v_\name{full}(t)}\) and \({T_{\name{err}} = \set{t \in D'}{v(t)=\neg v_\name{full}(t)}}\).
\end{definition}

Focusing on a subset~$D'$ of the database tuples will be useful later in Section~\ref{sec:mes_metric}. Note that both \(T_{\name{right}}\) and \(T_{\name{err}}\) can only contain tuples for which the label is known, i.e. \(v(t) \ne \bot\) for every tuple~$t$ in these subsets. Tuples with unknown labels have a probability~\(1\) of having any label, and therefore do not affect the computation.

\begin{example}
\label{example:labeling_prob}
Continuing Example~\ref{example:possible_worlds} and taking \(D' = D_{\name{ex}}\), we have \(T_{\name{right}} = \curly{a_1, a_2, a_4, e_2}\) and \(T_{\name{err}} = \curly{a_3, r_1}\). Therefore, \(P\parenth{\overline{D_{\name{ex}}}, v_\name{full}, D'} = (1 - \func{err}{a_1}) \cdot (1 - \func{err}{a_2}) \cdot (1 - \func{err}{a_4})\cdot (1 - \func{err}{e_2}) \cdot \func{err}{a_3} \cdot \func{err}{r_1}  = 0.01176
\).
\end{example}

\paragraph*{Queries over \des{}}
We now study how labeling errors propagate to query results, focusing on Select-Project-Join-Union (SPJU) queries~\cite{abiteboul1995foundations}. See Section~\ref{sec:mes_queries_generalization} for a discussion about other types of queries. The next definition shows how queries are applied over a \des{}, i.e., how the 3-valued correctness labels are defined for query tuples.
\begin{definition}[output correctness labeling]\label{def:output_correctness_labeling}
Let \(\overline{D}={\anglebr{D, v, \name{err}}}\) be a \des{} and \(Q\) an SPJU query over \(D\). For \(l \in \curly{\false, \true, \bot}\) denote \(D_l = \set{t \in D}{v(t) = l}\), i.e., the subset of~$D$ labeled as~\(l\) by~\(v\). An \emph{output correctness labeling} is a 3-valued function \(v^Q:Q(D)\to \curly{\false, \true, \bot}\) such that for every $o\in Q(D)$ we have \(v^Q(o)=\true\) if~\({o\in Q(D_{\true})}\), \(v^Q(o)=\false\) if~\({o\not\in Q(D_{\true}\cup D_{\bot})}\), and~\(v^Q(o)=\bot\) otherwise.
\end{definition}

Intuitively, an output tuple~\(o\) is correct iff it can be derived from correct \des{} tuples, and incorrect iff it is not derived even if all the unlabeled \des{} tuples are correct. This definition is applicable to SPJU queries because of their monotonicity. 

\begin{example}[Querying the DES]
\label{example:query_over_des}
Continuing Example~\ref{example:labeling_prob}, define \(Q_{\name{ex}}\) to be the SPJU query in Fig.~\ref{fig:query}, which selects acquired companies along with the universities from which the founding members have graduated before the acquisition. Table~\ref{tab:des_query_example} presents the results of \(Q_{\name{ex}}\) (ignore the \(L^{Q_{\name{ex}}}\)~column for now). The first output tuple can be derived from the input tuples~$a_1$, $r_1$, and~$e_2$, corresponding to the acquisition of BHealthy, having Emil Lime as its founder and as a graduate of U.\ São Paulo. Since these tuples are considered correct according to \(v_{\name{ex}}\), so is this query result. The third output tuple is considered incorrect, since its derivation requires the acquisition of Optobest (tuple~$a_4$), which is considered incorrect. The second output tuple is considered unknown since its correctness depends on input tuples~\(r_2\) and~\(e_1\), which are assigned to~\(\bot\).
\end{example}

\subsection{Provenance}
Our analysis requires computations over different possible worlds. For that, we will build on the notion of \emph{Boolean provenance}, which allows capturing how query output tuples are derived from input tuples independently of a specific ground truth. Indeed, Boolean provenance has been used for this purpose in many contexts~\cite{cheney2009provenance,green2007provenance,imielinski1984incomplete} and in particular in data cleaning~\cite{bergman2015query,drien2023query} and probabilistic and uncertain databases~\cite{suciu2011prob}. In this work, provenance is useful in the definitions of our uncertainty metrics and in our uncertainty reduction algorithm. As we are focusing on SPJU queries, we will use the $\mathrm{PosK_3}$ semiring, extending the \textrm{PosBool} semiring of~\cite{green2007provenance} with Kleene's 3-valued logic~$K_3$ rather than Boolean logic.\footnote{Kleene's~$K_3$ logic adds to the usual~$\false,\true$ an unknown value. Truth tables are then extended such that $\false\wedge\bot=\false$, $\bot\wedge\true=\bot$, etc.} This semiring allows deriving a logical expression for each output tuple, such that: \begin{inparaenum}[(i)] \item each variable corresponds to a related input tuple; \item the value of each variable corresponds to its correctness label (or its absence); and \item the computed value of an expression reflects the derived correctness of the corresponding output tuple, according to Def.~\ref{def:output_correctness_labeling}.\end{inparaenum}

For SPJU queries, provenance can be computed in PTIME as monotone~\(k\)-DNF formulae~\cite{imielinski1984incomplete}, i.e., disjunctions of conjunctions (terms) containing \emph{at most}~\(k\) variables each, without negation. For clarity, we assume each term has \emph{exactly}~\(k\) variables; but our results apply to the general case.

\begin{figure}
\setlength{\abovecaptionskip}{-5pt}
\begin{smaller}
\begin{Verbatim}[frame=single]
SELECT DISTINCT a.Acquired, e.University
FROM Acquisitions AS a, Roles AS r, Education AS e
WHERE a.Acquired = r.Company AND
      r.Name = e.Name AND
      r.Role ILIKE '%found%' AND
      e.Year <= DATE_PART('YEAR', a.Date)
\end{Verbatim}
\end{smaller}
\caption{Example query \(Q_{\name{ex}}\)}
\label{fig:query}
\end{figure}

\begin{table}[t]
  \caption{The results of \(Q_{\name{ex}}\) over \(D_{\name{ex}}\)}
  \small
  \centering
  \begin{tabular}{cc|c}
    \toprule
     Acquired & University & \(L^{Q_{\name{ex}}}\) \\
    \midrule
   BHealthy & U.\ São Paulo & \begin{smaller}\(L^{Q_{\name{ex}}}(o_1)=(a_1 \conj r_1 \conj e_2) \disj (a_1 \conj r_4 \conj e_3)\)\end{smaller} \\
     NewHealth & U. Melbourne & \begin{smaller}\(L^{Q_{\name{ex}}}(o_2)=(a_2 \conj r_2 \conj e_1) \disj (a_3 \conj r_2 \conj e_1)\)\end{smaller} \\
    Optobest & U.\ Cape Town & \begin{smaller}\(L^{Q_{\name{ex}}}(o_3)=a_4 \conj r_3 \conj e_4\)\end{smaller} \\
    \bottomrule
  \end{tabular}
  \label{tab:des_query_example}
\end{table}

\begin{definition}[An annotated \des{}]
\label{def:annotated_des}
An \emph{annotated \des{}} is a \des{} in which every tuple is equipped with a unique variable. Formally, it is a triplet \(\hat{D} = \anglebr{\overline{D}, X, L}\), where \(\overline{D} = \anglebr{D, v, \name{err}}\) is a \des{}, \(X\) is a set of \(\abs{D}\) ternary variables, and \(L\) is a bijection \(L: D \to X\).
\end{definition}

See, for example, column~$L$ in Table~\ref{tab:des_example}. The next definition shows how SPJU queries are applied over an annotated \des{}:

\begin{definition}
\label{def:annotated_des_query}
\sloppy Let \(\hat{D}\) be an annotated \des{} and \(Q\) be an SPJU query. The \emph{result of \(Q\) over \(\hat{D}\)} is the result over~\(D\) where every output tuple is equipped with a provenance expression. Formally, it is a triplet \(Q(\hat{D}) = \anglebr{Q(D), \Phi, L^Q}\), where \(Q(D)\) is the result of \(Q\) over \(D\), \(\Phi\subseteq\mathrm{PosK_3}[X]\) is a set of up to \(\abs{Q(D)}\) formulae over \(X\), and \(L^Q\) is a function that labels each output tuple in \(Q(D)\) by a formula in \(\Phi\) satisfying the following important property: for every output tuple \(o \in Q(D)\), \(v^Q (o) = v(L^Q (o))\), where \(v^Q\) is the output correctness labeling of Def.~\ref{def:output_correctness_labeling}, and \(v(L^Q (o))\) is the value of the~$K_3$ expression~\(L^Q (o)\) after substituting every variable~\(x\) in it by the correctness of its corresponding tuple \(v(L^{-1}(x))\).
\end{definition}

The key property described above reflects a general feature of provenance semirings -- namely, their commutation with semiring homomorphisms~\cite{green2007provenance}. Specifically, a 3-valued assignment to variables acts as a semiring homomorphism from expressions in $\mathrm{PosK_3}[X]$ to truth values in $K_3$. Intuitively, one can either (1) assign correctness labels to input tuples and then compute the query result using provenance to derive output labels, or (2) evaluate the query over provenance annotations to get abstract expressions and assign labels afterward. The second approach avoids re-evaluating the query when labels change, requiring only re-assignment of truth values.

\begin{example}
\label{example:provenance}
Continuing Example~\ref{example:query_over_des}, Table~\ref{tab:des_query_example} contains provenance expressions in 3-DNF form, in the \(L^{Q_{\name{ex}}}\)~column. For convenience, we denote the output tuples by~$o_1$, $o_2$, and~$o_3$. Notice that applying \(v_{\name{ex}}\) on the provenance expressions indeed yields the same labeling of \(v_{\name{ex}}^{Q_{\name{ex}}}\), as expected: \(v_{\name{ex}} \parenth{L^{Q_{\name{ex}}} (o_1)} = (\true \conj \true \conj \true) \disj (\true \conj \bot \conj \bot) = \true\), \(v_{\name{ex}} \parenth{L^{Q_{\name{ex}}} (o_2)} = (\false \conj \bot \conj \bot) \disj (\true \conj \bot \conj \bot) = \bot\), and \(v_{\name{ex}} \parenth{L^{Q_{\name{ex}}} (o_3)} = \false \conj \bot \conj \bot = \false\).
\end{example}

\section{Uncertainty Indicators}\label{sec:indicators}
In this section, we present and analyze our novel uncertainty indicators, starting with the uncertainty metric, \emph{Maximal Error Score (MES)}. These indicators can provide valuable insights to analysts and may also guide subsequent data verification (as we do in Section~\ref{sec:uncertainty_red_algs}).

\subsection{Maximal Error Score (MES) Metric}
\label{sec:mes_metric}
Uncertainty can propagate from input data to query results, as reflected in our running example: for instance, in Table~\ref{tab:examplenumbers}, the labeling~$v_{\name{ex}}$ marks the first output tuple~$o_1$ in Table~\ref{tab:des_query_example} as correct, while it is incorrect according to the ground truth labeling~$v_\name{full}$. This error is caused by an erroneous label on the input tuple~$r_1$. While the uncertainty of the input tuples is generally easily obtainable, at least in an approximate form (Section~\ref{sec:model}), the challenge is in \emph{quantifying} the uncertainty of the query output. We start by focusing on a single output tuple, and in Section~\ref{sec:multiple_tuples_mes} we generalize to multiple tuples or full query results.

\subsubsection{MES for a single output tuple}
To capture the uncertainty in the label of a query output tuple~$o$, we take a worst-case perspective: among all the possible worlds in which the label of~$o$ is erroneous, we identify the one where the observed input labels are most likely. The MES is defined as the likelihood of the observed input labels in this worst-case world. This provides an upper bound on how likely it is that the label of~$o$ is wrong -- across \emph{all} worlds where an error in~$o$ occurs.

A low MES value provides a \emph{strong guarantee} that the correctness label of an output tuple~\(o\) is reliable, because the observed verification labels are improbable in every possible world in which~\(o\)'s label is erroneous. Conversely, a high MES indicates that there exists a possible world where the observed labels are likely, yet the label of~\(o\) is wrong. 

\begin{definition}[Maximal Error Score (MES)]
\label{def:mes}
Let \(\hat{D} = \anglebr{\overline{D}, X, L}\) be an annotated \des{}, \(Q\)~be an SPJU query, and \(o \in Q(D)\) be an output tuple with a known correctness label. The \emph{Maximal Error Score (MES) of~\(o\)} is the maximal labeling probability over all the possible worlds under which~\(o\) is labeled erroneously. Formally,
\begin{equation}\label{eq:mes}
\scalebox{0.95}[1.0]{%
$\displaystyle
\func{MES}{\hat{D}, Q, o} = \max_{v_\name{full} \in \func{Inv}{\hat{D}, Q, o}} P\parenth{\overline{D}, v_\name{full}, \func{Rel}{\hat{D}, o}}
$%
}
\end{equation}
where \(\func{Inv}{\hat{D}, Q, o}\) is a set of possible worlds defined by
\begin{equation}\label{eq:inv}
\scalebox{0.87}[1.0]{%
$\displaystyle
\func{Inv}{\hat{D}, Q, o} = \set{v_\name{full}: D \to \curly{\false, \true}}{v_\name{full}(L^Q (o)) = \neg v^Q (o)}
$%
}
\end{equation}
\(P\parenth{\overline{D}, v_\name{full}, \func{Rel}{\hat{D}, o}}\) is the labeling probability defined in Eq.~\ref{equation:labeling_prob}, and \(\func{Rel}{\hat{D}, o}\) is the subset of~\(D\) with the input tuples that participate in the derivation of~\(o\), defined by
\begin{equation}\label{eq:rel}
\func{Rel}{\hat{D}, o} = \set{t \in D}{L(t) \in \vars{L^Q (o)}}
\end{equation}
where \(\vars{\phi}\) is the set of variables in \(\phi\). We refer to possible worlds that achieve the maximum in Eq.~\ref{eq:mes} as \emph{worst-case possible worlds}.
\end{definition}

\begin{example}
\label{example:mes}
Continuing Example~\ref{example:provenance}, we will illustrate MES computation for $o_1$ with respect to $v_{\name{ex}}$ in Table~\ref{tab:examplenumbers}. Note that among the annotations in $o_1$'s provenance, only~$a_1,r_1$ and $e_2$ have known labels, and that the derived correctness label of~$o_1$ is~\true{}; therefore, $\name{Inv}(\hat{D}, Q, o_1)$ consists of all the possible worlds in which~$o_1$ is incorrect. For instance, with the possible world~$v_{\name{full}}$  from Table~\ref{tab:examplenumbers}, the labeling probability restricted to~$\func{Rel}{\hat{D}, o_1)}$ is \((1 - \func{err}{a_1}) \cdot \func{err}{r_1} \cdot (1 - \func{err}{e_2}) = 0.7\cdot 0.2\cdot 0.6=0.084\). For a possible world $v'_{\name{full}}$ for which $v_{\name{ex}}$~errs only for~$e_2$, the labeling probability is \((1 - \func{err}{a_1}) \cdot (1-\func{err}{r_1}) \cdot \func{err}{e_2} = 0.7\cdot 0.8\cdot 0.4=0.224\). It can be shown that $v'_{\name{full}}$ is the worst-case possible world for~$o_1$ and $v_{\name{ex}}$, i.e., that it achieves the maximal probability compared with any possible world in $\name{Inv}(\hat{D}, Q, o_1)$. Hence, $\func{MES}{\hat{D}, Q, o_1} = 0.224$.

As another example, the MES of output tuple~$o_3$ with respect to~$v_{\name{ex}}$ is~0: according to Table~\ref{tab:examplenumbers}, the probability of an error in the label of~$a_4$ is~0. Hence, in every possible world where~$a_4$ is correct, the probability of the observed labeling~$v_{\name{ex}}$ is~0. However, the derived label of~$o_3$ as incorrect can only be erroneous if~$a_4$ is correct. In this case, a MES score of~0 corresponds to the intuition that there is a zero probability of an error in the label of~$o_3$, regardless of the ground truth.
\end{example}

Notice that, as opposed to previous probabilistic or possibilistic approaches to uncertain databases, computing the MES does not require prior knowledge (e.g., correctness probabilities) of input tuples, domain knowledge that may be unavailable or unaffordable. In addition, it obtains the worst-case perspective, which results in a strong guarantee of reliability for analysts.

\begin{example}
\label{example:mes_usage}
Recall from Example~\ref{ex:analyst} that even after applying a verifier (even an expensive, high-quality one such as an LLM or a human expert), an analyst still faces fundamental questions about the reliability of query results and the risk of missing correct output tuples.

The MES metric addresses these concerns by assigning each output tuple a worst-case uncertainty score. Output tuples with low MES are highly reliable, providing strong guarantees about whether they are correct query results. These scores also offer practical guidance to analysts: for instance, they can highlight the most uncertain output tuples for closer inspection or allow the analyst to retain only sufficiently reliable outputs for downstream analysis (as demonstrated in~\cite{schreiber2024cleaner}). Importantly, MES can be reported independently of the underlying verification mechanism and thus can be used on top of existing verification systems, without requiring adoption of the full framework shown in Figure~\ref{fig:architecture_diagram}.

How to further improve result reliability under a limited verification budget is addressed in Sections~\ref{sec:mes_metric_risky_tuples} and~\ref{sec:uncertainty_red_algs}, which build on the MES metric. Section~\ref{sec:case_studies} presents end-to-end case studies that illustrate these ideas in practice.
\end{example}

In some cases, e.g., when outputting it in the GUI of a system, it may be convenient to work with the averaged value
\begin{equation}
\label{equation:averaged_mes}
\func{\overline{MES}}{\hat{D}, Q, o} = \func{MES}{\hat{D}, Q, o}^{-\frac{1}{n}}
\end{equation}
where \(n\) equals the number of factors in the product, similar to the perplexity measure. In addition, it may also be useful to present the log of the MES values. See also Section~\ref{sec:mesreduce_optimizations}.

\subsubsection{MES for multiple tuples}\label{sec:multiple_tuples_mes}
Def.~\ref{def:mes} computes an uncertainty score for each query result. However, in some scenarios -- such as the one in Section~\ref{sec:uncertainty_red_algs} -- a single aggregated uncertainty score for the entire query output or a subset of its tuples may be desired. We propose, as a natural extension to Def.~\ref{def:mes} that preserves the worst-case perspective of MES as an uncertainty bound, to compute the MES of multiple output tuples as the maximal MES over these tuples: given an annotated database $\hat{D}$, a query~$Q$ and a set of output tuples $O\subseteq Q(D)$, we define,
\begin{equation}\label{eq:multiple-tuples-mes}
\func{MES}{\hat{D}, Q, O} = \max_{o\in O} \func{MES}{\hat{D}, Q, o}
\end{equation}
Indeed, if the maximal MES of tuples in~$O$ is low, then the uncertainty of every output tuple~$o\in O$ is correspondingly low, and we have a guarantee of reliability for every examined output tuple. On the contrary, if the maximal value is high, then there exists a highly uncertain tuple. We show this aggregation to be useful when we consider lowering the uncertainty in Section~\ref{sec:uncertainty_red_algs}, and when we examine the empirical performance of such algorithms in Section~\ref{sec:experiments}. Other aggregations are possible, see Section~\ref{sec:conclusions}.

\subsubsection{Generalizing beyond SPJU queries}\label{sec:mes_queries_generalization}
Our depiction of the formal framework can guide the generalization of the MES metric to additional query types. As an example, we outline the necessary steps for adapting the MES metric to \emph{aggregate queries}. First, a suitable provenance system should be adopted for aggregate queries, e.g.~\cite{amsterdamer2011provenance}; in particular, it is used for identifying related input tuples that participate in the derivation of the output tuple of interest (Eq.~\ref{eq:rel}). Then, the set \(\func{Inv}{\hat{D}, Q, o}\) (Eq.~\ref{eq:inv}) should be redefined, e.g., to one where the true aggregation result sufficiently differs from the result computed over tuples labeled as correct by the verifier.

For more complex downstream tasks (e.g., machine learning), SPJU queries can be used to identify and extract relevant input data (e.g., features or training instances). The MES metric can be applied to assess and filter out uncertain outputs of such queries, thereby improving the reliability of the downstream input. Further optimization of verification that accounts for the specific downstream task is a  direction for future work; see Section~\ref{sec:conclusions}.

\subsection{Risky Tuples}
\label{sec:mes_metric_risky_tuples}

A counterintuitive phenomenon when adopting a worst-case perspective of verification is that lowering the error probability of an input label can actually increase the uncertainty of an output tuple. We call such input tuples \emph{risky tuples}. We start with an example illustrating this phenomenon, then we formally define risky tuples, and lastly, we give a lemma of properties that is used in Section~\ref{sec:riskycomputation} for computationally identifying such tuples.

\begin{example}
\label{example:risky_tuple}
Continuing Example~\ref{example:mes}, define \(\name{err'}\) to be identical to \(\name{err}\) in Table~\ref{tab:examplenumbers} except that \(\func{err'}{a_1} = 0.1 < \func{err}{a_1}\). Denote by \(\hat{D'}\) the corresponding annotated \des{}. Both \(\overline{D}\) and \(\overline{D'}\) represent the same verification result, but in \(\overline{D'}\) we are more confident in the label of~\(a_1\). Surprisingly, the MES of~\(o_1\) has increased from \(0.224\) in~\(\hat{D}\) to \(0.288\) in~\(\hat{D'}\), i.e. the worst-case uncertainty has increased. Even more surprisingly, the other variables in the provenance of~\(o_1\) all exhibit this property: reducing \(\func{err}{r_1}\) to~\(0.1\) while keeping the other error probabilities increases the MES value to~\(0.252\), and reducing \(\func{err}{e_2}\) to~\(0.01\) increases the MES value to~\(0.237\).
\end{example}

When we reduce the error likelihood of a label for a risky tuple, we increase the likelihood of the labeling in a possible world; if this is, or becomes, the worst-case possible world, the MES will increase accordingly. Formally,

\begin{definition}[Risky input tuples]
\label{def:risky_tuple}
Let \(\hat{D} = \anglebr{\overline{D}, X, L}\) be an annotated \des{}, \(Q\)~be an SPJU query, and \(o \in Q(D)\) be an output tuple with a known correctness label (\(v^Q (o) \neq \bot\)). Let~\(t_0\) be an input tuple with a positive error probability \(\func{err}{t_0}\). A \(p \in [0, \func{err}{t_0})\) is called an \emph{unsafe error probability for~\(t_0\)} if reducing the uncertainty to~$p$ results in a higher MES value for~\(o\), i.e. \(\func{MES}{\hat{D'}, Q, o} > \func{MES}{\hat{D}, Q, o}\), where \(\hat{D'} = \anglebr{\overline{D'}, X, L}\) and \(\overline{D'} = \anglebr{D, v, \name{err'}}\) with \(\func{err'}{t_0} = p\) and \(\func{err'}{t} = \func{err}{t}\) for every other tuple~\(t\). \(t_0\)~is called a \emph{risky input tuple} if it has an unsafe error probability. If every probability \(p \in [0, \func{err}{t_0})\) is unsafe, it is called a \emph{reliability-impairing input tuple}. If there is no unsafe error probability for it, it is called a \emph{safe input tuple}.
\end{definition} 

Besides the theoretical interest in the phenomenon, the indication of risky input tuples can practically help the data analysts determine how to invest further verification efforts, and can be leveraged by automatic verification algorithms, e.g., as in Section~\ref{sec:uncertainty_red_algs}.

\begin{example}
\label{example:risky_tuples_usage}
Continuing the analyst scenario from Example~\ref{example:mes_usage}, suppose the analyst identifies several query output tuples that are critical for her task but have high MES, indicating high uncertainty. The analyst is willing to invest additional effort to improve reliability, for example, by sending some of the input tuples to a domain expert or an LLM for re-verification. The \emph{risky tuples} indicator intuitively means that the choice of tuples to re-verify is crucial, by identifying input tuples whose refinement could \emph{increase} overall output uncertainty. By highlighting such tuples, the indicator helps the analyst focus verification efforts on non-risky inputs that are more likely to improve reliability. In~\cite{schreiber2024cleaner}, we demonstrate an interactive graphical interface that enables analysts to inspect risky tuples and make informed refinement decisions.
\end{example}

Note that Def.~\ref{def:risky_tuple} decreases only a single probability. When multiple probabilities are decreased simultaneously, the MES value may decrease even if risky input tuples are involved. A trivial example is the decrease of all the error probabilities to zero, in which case a zero MES value is guaranteed. We use this observation in Section~\ref{sec:uncertainty_red_algs}.

The following example shows that a risky input tuple may not be a reliability-impairing tuple:

\begin{example}
Continuing Example~\ref{example:risky_tuple}, \(0.01\) is an unsafe error probability for \(e_2\), so \(e_2\) is a risky input tuple. However, \(0.1\) is not an unsafe error probability: it reduces the MES value to \(0.216\). Therefore, \(e_2\) is risky but not reliability-impairing.
\end{example}

The following lemma presents two basic properties of risky input tuples, which will later be useful for detecting them.

\begin{lemma}
\label{lemma:unsafe_probs_properties}
Let \(\hat{D} = \anglebr{\overline{D}, X, L}\) be an annotated \des{}, \(Q\)~be an SPJU query, and \(o \in Q(D)\) be an output tuple with a known correctness label. Let~\(t_0\) be a risky input tuple. Then,
\begin{enumerate}
    \item If~\(q\) is an unsafe error probability for~\(t_0\), then every \(q' \in [0, q)\) is also unsafe for~\(t_0\).
    \item There must exist a \emph{positive} unsafe error probability for~\(t_0\).
\end{enumerate}
\end{lemma}

\begin{proof}
\begin{proofpart}
\normalfont Define a function \(f: [0, \func{err}{t_0}] \to [0, 1]\) that maps every probability \(p\) to the MES value \(MES\parenth{\hat{D_p}, Q, o}\), where \(\hat{D_p} = \anglebr{\overline{D_p}, X, L}\) and \(\overline{D_p} = \anglebr{D, v, err_p}\) with \(err_p(t_0) = p\) and \(err_p(t) = \func{err}{t}\) for any other \(t\). For every \(v_\name{full} \in \func{Inv}{\hat{D}, Q, o}\) and every \(p \in [0, \func{err}{t_0})\), the probability \(P\parenth{\overline{D_p}, v_\name{full}, \func{Rel}{\hat{D}, o}}\) is a linear function of \(p\). Therefore, \(f\)~is a piecewise linear function.

By definition, \(f(q) > f(\func{err}{t_0})\), so there must exist some \(q_0 \in [q, \func{err}{t_0}]\) that lies in an interval \(I \subset [0, \func{err}{t_0}]\) in which the slope of \(f\) is negative. All slopes to the left of \(I\) must also be negative, so for every \(q' \in [0, q)\), \(f(q') > f(\func{err}{t_0})\) as desired.
\end{proofpart}

\begin{proofpart}
\normalfont Assume by contradiction that \(0\)~is the only unsafe probability for~\(t_0\). Using the same \(f\) as in the previous part, \(f(0) > f(\func{err}{t_0})\). Since \(f\)~is piecewise linear, it's continuous, so for every small~\(\epsilon > 0\), \(f(\epsilon) > f(\func{err}{t_0})\), and there must exist a positive unsafe error probability for~\(t_0\) in contradiction.\qedhere
\end{proofpart}
\end{proof}

As seen in Example~\ref{example:risky_tuple}, all input tuples participating in the derivation of an output tuple~\(o\) might be risky. The following example shows that this phenomenon may happen regardless of the correctness labels of~\(o\) or the input tuples:

\begin{example}
\label{example:all_tuples_risky}
 Returning to the annotated \des{}~\(\hat{D}\) presented in Table~\ref{tab:des_example}, define~\(Q_{\name{ex}2}\) to be the following query, which selects companies along with universities from which the members have graduated:
\begin{smaller}
\begin{Verbatim}[frame=single]
SELECT DISTINCT r.Company, e.University
FROM Roles AS r, Education AS e
WHERE r.Name = e.Name
\end{Verbatim}
\end{smaller}

\begin{table}[t]
  \small
  \centering
  \begin{tabular}{cccc}
    \toprule
    & Company & University & \(L^{Q_{\name{ex}2}}\) \\
    \midrule
    \(o_1\) & BHealthy & U.\ São Paulo & \begin{smaller}\((r_1 \conj e_2) \disj (r_4 \conj e_3)\)\end{smaller} \\
    \(o_2\) & NewHealth & U. Melbourne & \begin{smaller}\(r_2 \conj e_1\)\end{smaller} \\
    \(o_3\) & Optobest & U.\ Cape Town & \begin{smaller}\(r_3 \conj e_4\)\end{smaller} \\
    \bottomrule
  \end{tabular}
  \caption{The results of \(Q_{\name{ex}2}\) over \(\hat{D}_{\name{ex}}\)}
  \label{tab:resultsqex2}
\end{table}

Table~\ref{tab:resultsqex2} presents the results of~\(Q_{\name{ex}2}\) over \(\hat{D}_{\name{ex}}\). Define an assignment~\(v_1\) by \(v_1 (r_1) = v_1 (e_3) = \true\), \(v_1 (r_4) = v_1 (e_2) = \false\), and \(v_1 (x) = \bot\) for any other~\(x\), and define the error probabilities function~\(\name{err}\) by \(\func{err}{r_1} = \func{err}{r_4} = \func{err}{e_2} = \func{err}{e_3} = 0.3\), and \(\func{err}{x} = \bot\) for any other~\(x\). Notice that~\(o_1\) is labeled as incorrect. It can be verified that its MES value is~\(0.1029\) and that for every labeled input tuple -- both the tuples labeled as correct and the tuples labeled as incorrect -- decreasing the error probability to~\(0.2\) increases the MES value.

Now define an assignment~\(v_2\) by \(v_2 (r_1) = v_2 (e_2) = v_2 (e_3) = \true\), \(v_2 (r_4) = \false\), and \(v_2 (x) = \bot\) for any other~\(x\), and define \(\name{err}\) as above. Now, \(o_1\)~is labeled as correct. It can be verified that also in this case its MES value is~\(0.1029\), and for every labeled input tuple, decreasing the error probability to~\(0.2\) increases the MES value.
\end{example}

\section{Computing the Indicators}
\label{sec:computation}

In the previous section, we defined the MES metric and risky tuples; we now analyze their computation (Step~3 in Fig.~\ref{fig:architecture_diagram}). Clearly, computing MES naively is generally infeasible due to the exponential growth in possible worlds with respect to the number of labels. We show that, interestingly, the complexity of MES computation depends on the correctness label of the output tuple. If labeled as incorrect, the problem is tractable, and we provide a PTIME algorithm; if labeled as correct, the problem becomes intractable. We identify a tractable query class and reduce the general case to an integer linear program (ILP). Our solutions have been implemented and experimentally validated as practical in Section~\ref{sec:experiments}.

\subsection{Computing MES for Incorrect Output Tuples}\label{sec:computemesincorrect}
In the case of output tuples labeled as incorrect, it suffices to iterate over a small subset of \(\name{Inv}(\hat{D}, Q, o)\) that provably contains a worst-case possible world. It is proved formally in the next proposition, whose proof is by an analysis of Alg.~\ref{alg:incorrect_tuple_mes}.

\begin{algorithm}
\caption{Calculating MES of an incorrect output tuple}
\label{alg:incorrect_tuple_mes}
{\small
\begin{algorithmic}[1]
\Input \(\hat{D} = \anglebr{\overline{D}, X, L}\), SPJU query \(Q\), an output tuple \(o \in Q(D)\) labeled as incorrect (i.e. \(v^Q (o) = \false\)).
\Output \(\func{MES}{\hat{D}, Q, o}\)
\State \(\mathsf{maxProb} \gets 0\)
\State Calculate the provenance \(L^Q (o) \gets \kdnfi\)\label{line:incorrect_mes_provenance}
\For{\(i \gets 1\) to \(m\)}
    \State For every \(t \in D\), set \(v_\name{full}^i (t)\) to \(\true\) if \(L(t) \in \curly{a_j ^ i}_j\), \(v(L(t))\) if \(v(L(t)) \neq \bot\) or \(1\) otherwise.\label{line:incorrect_mes_vfull_def}
    \State \(p \gets P\parenth{\overline{D}, v_\name{full}^i, \func{Rel}{\hat{D}, o}}\)\label{line:incorrect_mes_labeling_prob}
    \If{\(p > \mathsf{maxProb}\)}
        { \(\mathsf{maxProb} \gets p\)}
        \EndIf
\EndFor
\State \Return \(\mathsf{maxProb}\)
\end{algorithmic}
}
\end{algorithm}

\begin{proposition}
\label{proposition:incorrect_tuple_mes_alg_correctness}
Let \(\hat{D} = \anglebr{\overline{D}, X, L}\) be an annotated \des{}, \(Q\)~be an SPJU query, and \(o \in Q(D)\) be an output tuple labeled as incorrect (i.e. \(v^Q (o) = \false\)). The value \(\func{MES}{\hat{D}, Q, o}\) can be calculated in polynomial time in~$\cardinality{\hat{D}}$.
\end{proposition}

\begin{proof}
We will prove the proposition by presenting a PTIME algorithm that calculates the MES value of~\(o\), depicted in Alg.~\ref{alg:incorrect_tuple_mes}. The algorithm calculates the provenance expression of~\(o\) in $k$-DNF form (line~\ref{line:incorrect_mes_provenance}), and iteratively examines \(m\)~possible worlds, \(v_\name{full}^1, \dotsc, v_\name{full}^m\), where~\(m\) is the number of terms in the provenance expression (line~\ref{line:incorrect_mes_vfull_def}). Intuitively, each~\(v_\name{full}^i\) is the possible world in which the \(i\)-th term in~\(L^Q (o)\) is satisfied, and the rest of the variables are labeled as in~\(v\) unless they have an unknown label, in which case they are labeled as~\true{}. For each possible world, the corresponding labeling probability is calculated (line~\ref{line:incorrect_mes_labeling_prob}), and the maximal probability among these is returned.

Each step of Alg.~\ref{alg:incorrect_tuple_mes} can be computed in PTIME in the data size, in particular since, as stated above, $k$-DNF provenance can be computed for SPJU query in PTIME~\cite{imielinski1984incomplete}. For each~\(i\), it holds that \(v_\name{full}^i \in \func{Inv}{\hat{D}, Q, o}\), since by construction term~$i$ is satisfied, whereas we assumed $v^Q(o)=\false$. Lemma~\ref{proposition:incorrect_tuple_mes_lemma} below proves that iterating over this small subset of possible worlds is indeed sufficient for finding a worst-case world.
\end{proof}

\begin{lemma}
\label{proposition:incorrect_tuple_mes_lemma}
Let \(\hat{D} = \anglebr{\overline{D}, X, L}\) be an annotated \des{}, \(Q\)~be an SPJU query, and \(o \in Q(D)\) be an output tuple labeled as incorrect (i.e. \(v^Q (o) = \false\)). Denote \(L^Q (o) = \kdnfi\) the k-DNF provenance expression of~\(o\), and for every \(i \in \range{1}{m}\) denote \(v_\name{full}^i (t)\) as in Alg.~\ref{alg:incorrect_tuple_mes}. Denote \(A = \range{v_\name{full}^1}{v_\name{full}^m}\). Then,
\begin{equation}
\func{MES}{\hat{D}, Q, o} = \max_{v_\name{full} \in A} P\parenth{\overline{D}, v_\name{full}, \func{Rel}{\hat{D}, o}}
\end{equation}
\end{lemma}

Intuitively, we construct each possible world such that it resembles the labeling~$v$ as much as possible with respect to tuples with known labels, with a minimal change to make one term evaluate to~\true{}, which is necessary to make~$o$ correct in contrast with its label in~$v^Q$. Any other possible world can only differ from~$v$ in more labels, and hence the probability of~$v$'s labels can only be lower.

\begin{proof}
First, notice that \(A \subseteq \func{Inv}{\hat{D}, Q, o}\), so the inequality \(\mathrm{l.h.s} \ge \mathrm{r.h.s}\) is immediate. We will prove the opposite inequality by showing that for every \(v_\name{full} \in \func{Inv}{\hat{D}, Q, o}\) there exists some \(v_\name{full}^{i_0} \in A\) such that \(P\parenth{\overline{D}, v_\name{full}, \func{Rel}{\hat{D}, o}} \le P\parenth{\overline{D}, v_\name{full}^{i_0}, \func{Rel}{\hat{D}, o}}\). Let \(v_\name{full} \in \func{Inv}{\hat{D}, Q, o}\), then it must satisfy some term \(i_0\) in the provenance expression \(L^Q (o)\). We will show that \(v_\name{full}^{i_0}\) satisfies the desired inequality.

If \(P\parenth{\overline{D}, v_\name{full}, \func{Rel}{\hat{D}, o}} = 0\), the inequality is trivial. Assume it's positive. Denote by~\(T\) the set of related input tuples that are labeled: \(T = \set{t \in D}{v(t) \neq \bot \conj L(t) \in \vars{L^Q (o)}}\). We will decompose~\(T\) into two disjoint sets \(T_1, T_2\), where \(T_1\) contains the tuples whose annotating variables are in the \(i_0\)-th term in the provenance expression of~\(o\), and \(T_2\) contains the rest of the tuples:
\[
T_1 = \set{t \in T}{L(t) \in \curly{a_j ^ {i_0}}_j} \quad T_2 = T \setminus T_1
\]

Also, for every tuple~\(t\) denote for convenience
\[
p(v_\name{full}, t) = \casesfunc{v(t) = v_\name{full} (t)}{1 - err(t)}{v(t) = \neg v_\name{full} (t)}{err(t)}
\]

Now, we can express the labeling probabilities using~\(p\):
\begin{gather*}
P\parenth{\overline{D}, v_\name{full}, \func{Rel}{\hat{D}, o}} = \prod_{t \in T_1} p(v_\name{full}, t) \cdot \prod_{t \in T_2} p(v_\name{full}, t) \\
P\parenth{\overline{D}, v_\name{full}^{i_0}, \func{Rel}{\hat{D}, o}} = \prod_{t \in T_1} p(v_\name{full}^{i_0}, t) \cdot \prod_{t \in T_2} (1 - err(t))
\end{gather*}

Since \(err(t) \in [0, 0.5]\) for every~\(t\), we get that \(P\parenth{\overline{D}, v_\name{full}, \func{Rel}{\hat{D}, o}} \le P\parenth{\overline{D}, v_\name{full}^{i_0}, \func{Rel}{\hat{D}, o}}\) as desired.
\end{proof}

\begin{algorithm}
\caption{Calculating MES of a correct output tuple}
\label{alg:correct_tuple_mes}
\begin{algorithmic}[1]
\Input \(\hat{D} = \anglebr{\overline{D}, X, L}\), \(Q\), an output tuple \(o \in Q(D)\) labeled as correct.
\Output \(\func{MES}{\hat{D}, Q, o}\)
\If{\(\mathbf{ExistsZeroError1Certificate}(\hat{D},Q,o)\) }
    {\Return 0}\label{line:correct_mes_zero_err_cert}
\EndIf
\State Calculate \(L^Q (o) \gets \kdnfi\) and denote \(\vars{L^Q (o)} = \curly{a_i}_{i=1}^n\)\label{line:correct_mes_provenance}
\State Denote \(p_i \gets err\parenth{L^{-1} (a_i)}\) for every \(a_i\) with \(v\parenth{L^{-1} (a_i)} \neq \bot\)\label{line:correct_mes_defs_start}
\State Denote \(I \gets \set{i \in [n]}{v\parenth{L^{-1} (a_i)} \neq \bot}\), \(J \gets \set{i \in I}{p_i > 0}\)
\State Calculate \(\beta \gets \sum_{i \in J} \parenth{v(a_i) \cdot \ln (p_i) + (1 - v(a_i)) \cdot \ln (1 - p_i)}\)
\State Define a vector \(\vec{c} = \rowvector{c}{n}^T\) by \(c_i \gets \ln (1-p_i)-\ln (p_i)\) for \(i \in J \conj v(a_i) = \true\), \(c_i \gets \ln (p_i)-\ln (1-p_i)\) for \(i \in J \conj v(a_i) = \false\) and \(c_i \gets 0\) otherwise.
\State Solve the following integer linear program over \(\vec{a} = \rowvector{a}{n}^T\) and denote by \(\xi\) its objective:\\
\parbox{\columnwidth}{%
\begin{gather}
\begin{cases}
\text{maximize }c^T \cdot a\\[-0.3em]
\begin{aligned}
\text{subject to }
\sum\nolimits_{j=1}^k a_j ^ i \le k - 1, & \quad i=1,\dotsc,m \\
a_i \in \{0{,}1\}, & \quad i \in J \\
a_i = \casesfunc{v(a_i) = \true}{1}{v(a_i) = \false}{0}, & \quad i \in I \setminus J \\
a_i = 0, & \quad i \notin I
\end{aligned}
\end{cases}
\end{gather}
}\label{line:correct_mes_ilp}
\State \Return \(\func{exp}{\xi + \beta}\)
\end{algorithmic}
\end{algorithm}

\subsection{Computing MES for Correct Output Tuples}\label{sec:computemescorrect}
As noted earlier, when the output tuple is labeled as correct, MES computation is intractable:

\begin{proposition}
\label{proposition:correct_tuple_mes_intractable}
Let \(\hat{D} = \anglebr{\overline{D}, X, L}\) be an annotated \des{}, \(Q\)~be an SPJU query, and \(o \in Q(D)\) be an output tuple labeled as correct (i.e. \(v^Q (o) = \true\)). If \(P \neq NP\), \(\func{MES}{\hat{D}, Q, o}\) cannot be calculated in polynomial time in~\(L^Q(o)\). This holds even if we assume that all the error probabilities are equal and that all the relevant input tuples have known correctness.
\end{proposition}

\begin{proof}
Assume by contradiction that there is a polynomial-time algorithm for calculating the MES value in the described setting. We will show a polynomial algorithm for the independent set problem, which is a well-known NP-hard problem~\cite{garey1979computers}:
\[
\name{INDSET} = \set{(G, l)}{
\parbox{12em}{
\(G\) contains an independent set of size \(l\) or more
}
}
\]

Let \((G, i)\) be an input for the problem, and denote \(G = (V, E)\) where \(V = \range{v_1}{v_n}\) and \(E = \range{e_1}{e_m}\).

We will first construct a relational database \(D\) that consists of a single relation \(edges\) that contains all the tuples \((v, v_1, \dotsc, v_n, 1) \in {[n]}^n \times \curly{1}\) where for every \(k\), the element \(v_k\) equals \(k\) if \(\curly{v_i, v_k} \in E\) and \(v < k\), and \(0\) otherwise.

Then we will define an SPJ query \(Q\), where \(x\) is the last element:
\begin{smaller}
\begin{Verbatim}[frame=single]
SELECT DISTINCT R.x FROM edges as R, edges as S
WHERE R.v1 = S.v
OR R.v2 = S.v
OR ...
OR R.vn = S.v
\end{Verbatim}
\end{smaller}

Denote by \(\hat{D}\) an annotation of \(D\). It's easy to see that \(Q(\hat{D})\) contains a single output tuple \(o\) whose provenance expression is the 2-DNF formula \(\phi = \bigdisji{\curly{u, v} \in E}{}{\parenth{u \conj v}}\), and the MES value is in the form of \(0.3^k \cdot 0.7^{2m - k}\).

The provenance of the single output tuple in \(Q(\hat{D})\) is \((v_1 \conj v_2) \disj (v_1 \conj v_3) \disj (v_2 \conj v_4) \disj (v_3 \conj v_5)\).

Define a \des{} \(\tilde{D} = (\hat{D}, v, err)\) with \(v(L(t)) = \true\) and \(err(t) = 0.3\) for all \(t \in D\). By our assumption, we can calculate \(\name{MES}(\tilde{D}, Q, o)\) in polynomial time. Then, we can extract \(k\) from the MES value in polynomial time, and we answer that \((G, l) \in \name{INDSET}\) iff \(k \ge l\).

It's easy to see that the reduction takes polynomial time. We will prove that \((G, l) \in \name{INDSET}\) iff \(k \ge l\). First, notice that there is a natural correspondence between independent sets for \(G\) and labeling functions \(V: D \to \curly{\false, \true}\) for which \(o \notin Q(D_V)\), because the corresponding assignments of such labelings don't satisfy \(\phi\). In addition, the size of an independent set equals \(k\) from the above form of the MES value.

Now, if \(G\) contains an independent set of size \(l\) or more, then the corresponding labeling function sends at least \(l\) variables to \(\true\), and \(k \ge l\). In the opposite direction, if \(k \ge l\), there exists a labeling function that sends at least \(l\) variables to \(\true\), and the corresponding independent set contains at least \(l\) vertices.
\end{proof}

We next propose two approaches for dealing with the intractability of MES computation for output tuples labeled as correct. First, we state an equivalent result to Prop.~\ref{proposition:incorrect_tuple_mes_alg_correctness} for subclasses of SPJU that admit CDNF provenance.

\begin{proposition}
\label{proposition:cdnf-correct-tractability}
Let \(\hat{D} = \anglebr{\overline{D}, X, L}\) be an annotated \des{}, \(Q\)~be an SJU or SPU query, and \(o \in Q(D)\) be an output tuple labeled as incorrect (i.e. \(v^Q (o) = \false\)). The value \(\func{MES}{\hat{D}, Q, o}\) can be calculated in polynomial time in~$\cardinality{\hat{D}}$.
\end{proposition}
For the subclasses of SPJU queries mentioned above, specifically join-free or projection-free queries, CNF provenance (a conjunction of disjunctions) can be computed in PTIME. In these cases, MES can be computed using a modified version of Alg.~\ref{alg:incorrect_tuple_mes}, where \true{} is replaced with \false{}, DNF terms (conjunctions) are replaced with CNF clauses (disjunctions), and the analysis follows a similar approach.

In the general case, we formulate an integer linear program such that the desired MES value is a tractable function of its solution. Alg.~\ref{alg:correct_tuple_mes} presents a PTIME reduction to an ILP solver. Intuitively, the provenance expression of~\(o\) is calculated in $k$-DNF form (line~\ref{line:correct_mes_provenance}) and is used to define the integer linear program (lines~\ref{line:correct_mes_defs_start}-\ref{line:correct_mes_ilp}), in which \(\vec{a} = \rowvector{a}{n}^T\) represents a possible world in \(\func{Inv}{\hat{D}, Q, o}\). Formally,

\begin{proposition}
\label{proposition:correct_tuple_mes}
Let \(\hat{D} = \anglebr{\overline{D}, X, L}\) be an annotated \des{}, \(Q\)~be an SPJU query, and \(o \in Q(D)\) be an output tuple labeled as correct (i.e. \(v^Q (o) = \true\)). There exists a PTIME algorithm with oracle access to an ILP solver that computes \(\func{MES}{\hat{D}, Q, o}\).
\end{proposition}

\begin{proof}
First, notice that \(I\) and \(J\) are nonempty.
For every labeling function \(V: D \to \curly{\false, \true}\) define \(\vec{r_V} \in \curly{0, 1}^n\) by
\[
(r_V)_i = \casesfunc{V\parenth{L^{-1} (a_i)} = \true}{1}{V\parenth{L^{-1} (a_i)} = \false}{0}
\]

Let \(V\) be a labeling function in \(Inv\parenth{D, Q, o}\) with a positive probability \(P\parenth{\tilde{D}, V, o} > 0\). For each \(i \in \parenth{I \setminus J}\) (if exists), \(V\parenth{L^{-1} (a_i)} = v (a_i)\).
\[
\begin{split}
&\ln \parenth{P\parenth{\tilde{D}, V, o}} = \ln \parenth{\prod_{i \in I} \casesfunc{V\parenth{L^{-1} (a_i)} \neq v(a_i)}{p_i}{V\parenth{L^{-1} (a_i)} = v(a_i)}{1-p_i}} \\
&= \ln \parenth{\prod_{i \in J} \casesfunc{V\parenth{L^{-1} (a_i)} \neq v(a_i)}{p_i}{V\parenth{L^{-1} (a_i)} = v(a_i)}{1-p_i}} \\
&= \sum_{i \in J} \casesfunc{V\parenth{L^{-1} (a_i)} \neq v(a_i)}{\ln (p_i)}{V\parenth{L^{-1} (a_i)} = v(a_i)}{\ln (1-p_i)} \\
&= \sum_{i \in J} \casesfunc{v(a_i) = \true}{(\ln (1-p_i) - \ln (p_i)) \cdot ((r_V)_i - 1) + \ln (1 - p_i)}{v(a_i) = \false}{(\ln (p_i) - \ln (1-p_i)) \cdot (r_V)_i + \ln (1 - p_i)} \\
&= \sum_{i \in J} \casesfunc{v(a_i) = \true}{(\ln (1-p_i) - \ln (p_i)) \cdot (r_V)_i}{v(a_i) = \false}{(\ln (p_i) - \ln (1-p_i)) \cdot (r_V)_i} \\
&\quad + \sum_{i \in J} \casesfunc{v(a_i) = \true}{\ln (p_i)}{v(a_i) = \false}{\ln (1-p_i)} \\
&= c^T \cdot r_V + \beta
\end{split}
\]

Define \(A\) to be the nonempty subset of \(Inv\parenth{D, Q, o}\) with positive probability, then
\[
\begin{split}
&\argmax_{V \in Inv\parenth{D, Q, o}}{P\parenth{\tilde{D}, V, o}} = \argmax_{V \in A}{P\parenth{\tilde{D}, V, o}} \\
&= \argmax_{V \in A}{\ln \parenth{P\parenth{\tilde{D}, V, o}}} = \argmax_{V \in A}{\parenth{c^T \cdot r_V + \beta}} \\
&= \argmax_{V \in A}{\parenth{c^T \cdot r_V}}
\end{split}
\]

It's easy to see that the integer linear program in Alg.~\ref{alg:correct_tuple_mes} solves this optimization (and that it's feasible), so \(\xi = \max_{V \in A}{\parenth{c^T \cdot r_V}} = \ln \parenth{\func{MES}{\tilde{D}, Q, o}} - \beta\),
and we get \(\func{MES}{\tilde{D}, Q, o} = e^{\xi + \beta}\), as desired.
\end{proof}

\subsection{Risky Tuple Computation}\label{sec:riskycomputation}
We will now discuss the computation of risky tuples. The following theorem gives an algorithmic characterization of risky tuples:

\begin{theorem}
\label{theorem:risky_tuples_algorithm}
Let \(\hat{D} = \anglebr{\overline{D}, X, L}\) be an annotated \des{}, \(Q\)~be an SPJU query, and \(o \in Q(D)\) be an output tuple with a known correctness label. Let~\(t\) be an input tuple with a positive error probability \(\func{err}{t} \in (0, 0.5]\). Then, \(t\)~is risky iff \(\func{MES}{\hat{D_0}, Q, o} > MES\parenth{\hat{D}, Q, o}\), 
where \(\hat{D_0} = \anglebr{\overline{D_0}, X, L}\) and \(\overline{D_0} = \anglebr{D, v, \name{err_0}}\) with \(\func{err_0}{t} = 0\) and \(\func{err_0}{t'} = \func{err}{t'}\) for every other tuple~\(t'\).
\end{theorem}

The "if" direction is immediate from the definition, and the "only if" direction derives from the first part of Lemma~\ref{lemma:unsafe_probs_properties}. Theorem~\ref{theorem:risky_tuples_algorithm} can be implemented as an algorithm to check if an input tuple is risky by performing two MES calculations: one with the current error probabilities and another in a hypothetical scenario where the error probability of the tuple is reduced to its minimum (0). MES computation is done using the algorithms described earlier. The algorithm can also be extended to a parameterized version, where we evaluate if reducing the error probability of tuple~$t$ to~$q>0$ results in a MES increase for~$o$. This extension is useful for assessing the risk of specific verification steps, especially when the post-verification error probability is known.

As discussed in Section~\ref{sec:computation}, MES calculation is an intractable problem. The following lemma presents two special cases in which an input tuple can be checked to be risky in polynomial time:

\begin{lemma}
Let \(\hat{D} = \anglebr{\overline{D}, X, L}\) be an annotated \des{}, \(Q\)~be an SPJU query, and \(o \in Q(D)\) be an output tuple with a known correctness label. Let~\(t_0\) be an input tuple with a positive error probability \(err(t_0) \in (0, 0.5]\).
\begin{enumerate}
    \item If there doesn't exist a possible world \(v_\name{full} \in \func{Inv}{\hat{D}, Q, o}\) such that \(v_\name{full} (t_0) = v(t_0)\), then \(t_0\)~is not risky.
    \item \sloppy If there exists a possible world \(v_\name{full} \in \func{Inv}{\hat{D}, Q, o}\) such that \(v_\name{full} (t_0) = v(t_0)\), but there doesn't exist a possible world \(v_\name{full} \in \func{Inv}{\hat{D}, Q, o}\) such that \(v_\name{full} (t_0) = \neg v(t_0)\) then \(t_0\)~is risky.
\end{enumerate}
\end{lemma}

\section{Uncertainty Reduction Algorithms}
\label{sec:uncertainty_red_algs}
In this section, we focus on settings where label error probabilities of input tuples can be reduced through additional verification, such as by invoking a more reliable verifier or by aggregating feedback from additional sources. The goal is to strategically select input tuples for further verification to obtain better worst-case reliability guarantees of query output; that is, to reduce the MES metric (Step~4 in Fig.~\ref{fig:architecture_diagram}).

We begin by introducing a generic algorithm scheme, \emph{MESReduce}, for reducing MES values in query outputs (Section~\ref{sec:mesreduce}). This scheme relies on four subroutines, whose possible implementations are described in Section~\ref{sec:uncertainty_red_algs_inst}. One may adapt the scheme to other settings by providing alternative implementations for them. Lastly, hardness results are presented in Section~\ref{sec:uncertainty_red_algs_hardness}.

\subsection{MESReduce Generic Scheme}\label{sec:mesreduce}
Our generic MESReduce scheme for lowering MES of query outputs is presented in Algorithm~\ref{alg:mes_reduction}, with subroutines marked by~(\(\star\)). As shown in Section~\ref{sec:uncertainty_red_algs_hardness}, selecting an optimal batch of tuples under a budget is intractable even without label changes, motivating an iterative approach. Selecting multiple tuples per iteration is necessary in the presence of risky tuples (Section~\ref{sec:mes_metric_risky_tuples}).

\paragraph*{Input}
The algorithm takes as input an annotated \des{}, an SPJU query, an optional target MES threshold, an optional cost budget, and a set of output tuples of interest~\(O\), along with provenance expressions that allow analysts to focus verification efforts on selected outputs.

\begin{algorithm}
\caption{MESReduce -- Reducing the maximal MES value}
\label{alg:mes_reduction}
\begin{algorithmic}[1]
\Input Annotated \des{} \(\hat{D} = \anglebr{\overline{D}, X, L}\), budget \(B \in (0, \infty)\),  MES threshold~$\Theta$, SPJU query~$Q$, output tuples~\(O \subseteq Q(\hat{D})\)
\State \(\mathsf{budgetLeft} \gets B\)
\While{\(\mathsf{budgetLeft} > 0\) and \(\name{MES}(\hat{D},Q,O) > \Theta\)}
    \If{$\exists o\in O.~v^Q (o) = \bot$}\label{line:mesreduce_reverify_cond}
    \State \((\star)~\hat{D}, \mathsf{budgetLeft} \gets \mathsf{ReVerify}(\hat{D},O,\mathsf{budgetLeft})\)\label{line:mesreduce_reverify}
    \EndIf
    \State \(o^* \gets \argmax_{o \in O}{\name{MES}(\hat{D}, Q, o)}\)\label{line:mesreduce_argmax}
    \State \((\star)~S\gets\mathsf{FindImprovementSet}(\hat{D},o^*,L^Q (o^*))\)\label{line:mesreduce_imp_set}
    \State \adjustbox{max width=0.94\linewidth}{\((\star)~p \gets \mathsf{NextProbability}(\hat{D},L^Q (o^*), S, \mathsf{budgetLeft}, \Theta)\)}\label{line:mesreduce_nextprob}
    \State \adjustbox{max width=0.94\linewidth}{\((\star)~\hat{D}, \mathsf{budgetLeft} \gets \mathsf{ImproveVerification}(\hat{D},S,p,\mathsf{budgetLeft})\)}\label{line:mesreduce_improveverification}
\EndWhile
\end{algorithmic}
\end{algorithm}

\paragraph*{Initial re-verification}

At each iteration, output tuples in~$O$ must have known correctness labels. If any are missing, the $\mathsf{ReVerify}$ subroutine (lines~\ref{line:mesreduce_reverify_cond}-\ref{line:mesreduce_reverify}) is invoked to generate them using an external verification tool. This step is repeated each iteration, as updates to input labels may invalidate derived output labels; for example, when changes in a provenance expression require previously unknown labels to be resolved. For example, if some $o \in O$ has provenance $x \vee y$, $v(y)=\bot$, and $x$ changes from~\true{} to~\false{}, we must obtain a label for~$y$ to determine $o$'s correctness.

\paragraph*{Selecting tuples to verify}
The algorithm identifies the least certain (highest MES) output tuple~$o^* \in O$ (line~\ref{line:mesreduce_argmax}), and selects a subset~$S$ of its provenance variables called an \emph{improvement set} via the $\mathsf{FindImprovementSet}$ subroutine (line~\ref{line:mesreduce_imp_set}). We thus adopt Eq.~\ref{eq:multiple-tuples-mes} for optimizing the reliability of the set~\(O\). Variants may consider multiple high-MES tuples, trading off fewer iterations against adaptability (see Section~\ref{sec:hyperparams_experiments}).

\paragraph*{Executing verification steps}
Verification is performed using an external tool, which may differ from the one used in $\mathsf{ReVerify}$. If the verifier supports controlling error probabilities, the $\mathsf{NextProbability}$ subroutine (line~\ref{line:mesreduce_nextprob}) is used to set target error probabilities for the tuples of~$S$. The verifier is then applied to these tuples using the $\mathsf{ImproveVerification}$ subroutine (line~\ref{line:mesreduce_improveverification}).

\begin{example}
\label{example:mesreduce_usage}
Continuing Example~\ref{example:risky_tuples_usage}, suppose the analyst seeks strong reliability guarantees for the query results but operates under a limited verification budget. Rather than manually selecting tuples to refine, she can invoke the \emph{MESReduce} algorithm to interact with the verifier on her behalf. In practice, this means that MESReduce automatically selects which input tuples to send for re-verification, for example, by placing them in a queue for a domain expert or issuing verification requests to an LLM-based verifier. Budget constraints naturally arise in such settings: the analyst may only be able to afford a fixed amount of expert time or a limited number of LLM queries. Under these constraints, MESReduce prioritizes verification actions that will reduce output uncertainty, thereby improving query result reliability more effectively than uninformed or uniform selection of tuples. Moreover, MESReduce can terminate early, before exhausting the available budget, once a desired reliability level is achieved.
\end{example}

\paragraph*{Generalizing to multiple queries}
MESReduce can be extended to multiple queries with different budgets or stopping criteria by annotating each output tuple with the disjunction of its provenance expressions across queries and tracking budgets and criteria per tuple.

\subsubsection{Optimizations for MESReduce}
\label{sec:mesreduce_optimizations}
The generic algorithm described above already admits several useful optimizations:
\begin{enumerate}
    \item Verifying the tuples in improvement sets may be done in parallel.
    \item Recalculation of MES values only needs to be done for tuples affected by the verification steps of the previous iteration, and can be executed in parallel; the affected tuples can be efficiently identified by analyzing the provenance.
    \item To avoid underflow, the algorithms can work with the log of MES values, adjusting the algorithms in Section~\ref{sec:computation} accordingly.
\end{enumerate}

\subsection{Subroutines Implementation}
\label{sec:uncertainty_red_algs_inst}
We next discuss implementations of the MESReduce subroutines, which may vary depending on the setting.

\paragraph*{External data verification tools}
The $\mathsf{ReVerify}$ and $\mathsf{ImproveVerification}$ subroutines (lines~\ref{line:mesreduce_reverify} and~\ref{line:mesreduce_improveverification}) are implemented by invoking external verifiers.
For $\mathsf{ReVerify}$, the output of the tool must be converted into a \des{}, including correctness labels and error probabilities. Section~\ref{sec:model} discusses how such probabilities can be obtained in practice, and Section~\ref{sec:case_studies} presents case studies with state-of-the-art tools. The $\mathsf{ImproveVerification}$ subroutine requires a verifier that can increase label accuracy for selected tuples, either with fixed cost and improvement (e.g., switching to a more accurate model) or with cost-dependent improvement, as in crowdsourcing, where increasing the number of workers reduces error.

\paragraph*{Improvement sets choice}
We now present an implementation of the $\mathsf{FindImprovementSet}$ subroutine, which selects an efficient batch of input tuples \(S\) for verification. This approach leverages risky-tuple identification and the structure of the provenance expression for the current output tuple~$o^* \in O$. If a related safe tuple with positive error probability exists, $S$ includes such a tuple with minimal verification cost. Otherwise, $S$ is defined based on $v^Q(o^*)$: if $v^Q(o^*)=\true$, $S$ contains the tuples in the cheapest satisfied term of the provenance expression~$L^Q(o^*)$; if $v^Q(o^*)=\false$, $S$ is the cheapest set covering at least one unsatisfied variable per term. Provenance expressions are computable in PTIME as monotone~$k$-DNF formulae~\cite{imielinski1984incomplete}. In the former case, $\abs{S}\leq k$; in the latter, $\abs{S}$ is bounded by the number of terms, and while optimal selection is tractable for CDNF provenance (Section~\ref{sec:computemescorrect}), it is intractable in general, requiring an ILP solver or a Hitting Set approximation.

The following example illustrates that for some formulae, \(T\) must be of size \thetaof{\abs{Vars(\psi)}}, where \(\phi\) is the provenance expression of the output tuple:

\begin{example}
Define \(\phi = \bigdisji{i = 1}{n}{\parenth{a_{2i} \conj a_{2i + 1}}}\) to be the provenance expression of the current output tuple, and \(v \equiv f\). It's easy to construct a database and a query that achieve this provenance expression. For every set of error probabilities \(\curly{p_i}\), when calculating the MES, it suffices to look only at ground truth assignments that satisfy exactly two variables. Their probabilities are in the form of \(p_j p_{j + 1} \prod_{i \notin \curly{j, j + 1}} \parenth{1 - p_i}\). If we start with \(p_1 = \dotsc = p_n\), the MES equals \(p^2 (1-p)^{n - 2}\), and \(T\) cannot contain less than \(n = \mathrm{\Theta}\parenth{\abs{Vars(\phi)}}\) variables.
\end{example}

\begin{lemma}
Unless \(\text{P} = \text{NP}\), the problem of finding the cheapest improvement set according to the procedure above, when a relevant safe tuple doesn't exist, cannot be solved in polynomial time.
\end{lemma}

\begin{proof}
\sloppy We will show a reduction from the hitting set problem~\cite{garey1979computers}. Let \(C = \range{S_1}{S_m}\) be a collection of subsets of a set \(S = \range{s_1}{s_n} \subset \mathbb{N}\) and let \(k\) be a positive number. We need to determine if a hitting set of size at most~\(k\) exists. We can assume that the sets \(S_i\) are equipotent and each contains \(l \ge 2\) elements. If this is not the case, we will first reduce the general hitting set problem to this restricted problem by adding new elements.

Denote by \(\hat{D}\) an annotated DES and \(Q\) a query such that \(Q(\hat{D})\) contains a single output tuple \(o\) whose provenance expression is the l-DNF formula \(\phi = \bigdisji{i=1}{m}{\parenth{\bigconji{s_j \in S_i}{}{s_j}}}\). Define a \des{} \(\hat{D} = \anglebr{\overline{D}, X, L}\) where \(\overline{D} = \anglebr{D, v, err}\) with \(v(L(t)) = \false\) and \(err(t) = 0.5\) for all \(t \in D\). Assume equal verification costs, so picking the cheapest set is equivalent to picking the smallest set.

Every input variable is risky: Let \(s_i\) be a variable, for every satisfying assignment for \(\psi\) the score equals \(\frac{1}{2^{n - 1}} \cdot p_i\) or \(\frac{1}{2^{n - 1}} \cdot (1 - p_i)\), and since no term solely contain \(s_i\) there exists a satisfying assignment for \(\psi\) which doesn't satisfy \(s_i\), so its score is \(\frac{1}{2^{n - 1}} \cdot (1 - p_i)\) and \(s_i\) is risky.

Finally, since the smallest selection that covers every term is a hitting set, we can determine if a hitting set of size at most~\(k\) exists iff the smallest selection contains at most~\(k\) variables.
\end{proof}

Note that improvement sets computed by the above procedure include only tuples with \emph{known labels}. Based on the initial information we have on these tuples, we can prove, as stated by the following proposition, that the above-defined improvement sets can indeed lower the MES. Formally,

\begin{proposition}
\label{proposition:mes_lowering_probs}
Let \(\hat{D} = \anglebr{\overline{D}, X, L}\) be an annotated \des{}, \(Q\) be an SPJU query, and \(o \in Q(D)\) be an output tuple with a known correctness label. Denote by \(T\) the set of input tuples chosen according to the procedure above. Then, there exists a sufficiently low probability \({p > 0}\) such that \(\func{MES}{\hat{D'}, Q, o} < MES\parenth{\hat{D}, Q, o}\), where \(\hat{D'} = \anglebr{\overline{D'}, X, L}\) and \(\overline{D'} = \anglebr{D, v, \name{err'}}\) for every error probabilities \(\name{err'}\) with \(\func{err'}{t_i} \le p\) for every \(t_i \in T\), and \(\func{err'}{t} = \func{err}{t}\) for every other tuple \(t\).
\end{proposition}

\begin{proof}
By definition, \(MES\parenth{\hat{D}, Q, o}\) is the maximal labeling probability over \(\func{Inv}{\hat{D}, Q, o}\). Since Lemma~\ref{lemma:mes_lowering_probs} is applicable for every such probability, the desired inequality is achieved for it, and therefore it is achieved for the maximal probability as well.
\end{proof}

\begin{lemma}
\label{lemma:mes_lowering_probs}
Let \(p_1, \dotsc, p_n \in \parenth{0, \frac{1}{2}}\) and let \(b_1, \dotsc, b_n \in \curly{0, 1}\). Assume that at least \(l > 0\) \(b_i\)-s are \(1\), and let \(\alpha_1, \dotsc, \alpha_n \in \parenth{0, 2^{\frac{l-n}{l}}}\). Then,
\[
\prod_{i = 1}^n (\alpha_i p_i)^{b_i} (1-\alpha_i p_i)^{1-b_i} < \prod_{i = 1}^n {p_i}^{b_i} (1-p_i)^{1-b_i}
\]
\end{lemma}

\begin{proof}
\sloppy Denote by \(l' \ge l\) the exact number of \(b_i\)-s that are \(1\). Assume w.l.o.g that the first \(l'\) \(b_i\)s are \(1\), and the rest are \(0\). If \(l' = n\), the claim is trivial. Otherwise, \(\prod_{i = 1}^n (\alpha_i p_i)^{b_i} (1-\alpha_i p_i)^{1-b_i} = \prod_{i = 1}^{l'} \alpha_i p_i \cdot \prod_{i = l' + 1}^n (1-\alpha_i p_i) < \prod_{i = 1}^{l'} \alpha_i p_i < \prod_{i = 1}^{l'} 2^{\frac{l-n}{l}} p_i = 2^{\frac{l' (l - n)}{l}} \cdot \prod_{i = 1}^{l'} p_i\). Notice that \(\prod_{i = l' + 1}^n (1-p_i) > \prod_{i = l' + 1}^n \frac{1}{2} = 2^{l' - n}\), and because \(\frac{l' (l-n)}{l} \le \frac{l (l-n)}{l} = l - n \le l' - n\), we get \(2^{\frac{l' (l - n)}{l}} \cdot \prod_{i = 1}^{l'} p_i < 2^{l' - n} \prod_{i = 1}^{l'} p_i < \prod_{i = 1}^{l'} p_i \cdot \prod_{i = l' + 1}^{n} (1 - p_i) = \prod_{i = 1}^n {p_i}^{b_i} (1-p_i)^{1-b_i}\) as desired.
\end{proof}

Proposition~\ref{proposition:mes_lowering_probs} is demonstrated in the following example:

\begin{example}
If \(p_1 = p_2 = 0.2, p_3 = p_4 = 0.3, b_1 = b_3 = 1\), and \(b_2 = b_4 = 0\), then \(\prod_{i = 1}^4 {p_i}^{b_i} (1-p_i)^{1-b_i} = 0.2 \cdot (1 - 0.2) \cdot 0.3 \cdot (1 - 0.3) = 0.0336\). According to the lemma, if we take \(l = 1\), we can multiply every probability by \(0.125\) to get \(p_1 = p_2 = 0.025, p_3 = p_4 = 0.0375\), and the product would be lowered to \(0.025 \cdot (1 - 0.025) \cdot 0.0375 \cdot (1 - 0.0375) = 0.0008\). If we take \(l = 2\), we can multiply every probability by \(\frac{1}{2}\) to get \(p_1 = p_2 = 0.1, p_3 = p_4 = 0.15\), and the product would be lowered to \(0.1 \cdot (1 - 0.1) \cdot 0.15 \cdot (1-0.15) = 0.011475\).
\end{example}

\paragraph*{Error probabilities choice}
There are several ways to implement the \(\mathsf{NextProbability}\) procedure when the external tools support adjusting target error probabilities (e.g., via crowd size in crowdsourcing or computation time in AI models). We propose a strategy inspired by variable step size in optimization (e.g.,~\cite{shalev2014understanding, patterson2017deep}):
\begin{compactenum}
    \item A monotonically decreasing function is set, e.g. \(f(n) = \frac{1}{n}\).
    \item Given an output tuple~$o^*$, the current minimal non-zero error probability~\(q\) of an input tuple that participates in the derivation of the current output tuple is taken, and we set \(n = \ceil{\frac{1}{q}}\).
    \item The verification probability is defined to be \(\max\{f(n + 1),\Theta\}\), where~$\Theta$ is the optional target MES. Intuitively, we aim at selecting a probability slightly lower than~$q$, but not lower than~$\Theta$.
\end{compactenum}

\subsection{MES reduction hardness}
\label{sec:uncertainty_red_algs_hardness}

The following proposition shows that reducing the MES value of a formula labeled as~\false{} is a computationally hard problem:

\begin{proposition}
Let \(\phi\) be a \(k\)-DNF formula, \(v\) be an assignment such that \(v(\phi) = \false\), and \(\curly{p_i}\), $q$ be probabilities. If \(P \neq NP\), the minimal selection of variables for which updating their probabilities to~$q$ would reduce the MES value cannot be done in polynomial time.
\end{proposition}

\begin{proof}
We will show that the problem is computationally hard in the case where the probabilities are all \(\frac{1}{2}\) and \(v \equiv \false\).

We will show a reduction from the hitting set problem~\cite{garey1979computers}. Let \(C = \range{S_1}{S_m}\) be a collection of subsets of a set \(S = \range{s_1}{s_n} \subset \mathbb{N}\) and let~\(l\) be a positive number. We need to determine if a hitting set of size at most~\(l\) exists. We can assume that the sets~\(S_i\) are equipotent and each contains at least two elements. If this is not the case, we will first reduce the general hitting set problem to this restricted problem by adding new elements.

Define \(\phi = \bigdisji{i=1}{m}{\parenth{\bigconji{s_j \in S_i}{}{s_j}}}\). Notice that a selection of variables to improve would lower the MES value iff it is a hitting set for the terms, assuming that~\(q\) is low enough. For example, if~\(l\) variables are improved while leaving an unimproved term, the resulting MES value would be \((1-q)^l \cdot 2^{l-n}\). Therefore, if such a selection can be computed in PTIME, then a PTIME solution of the hitting set problem would result.
\end{proof}

The following proposition shows that reducing the maximal MES value of multiple formulae labeled as~\true{} is a computationally hard problem:

\begin{proposition}
Let \(\phi_1, \dotsc \phi_m\) be \(k\)-DNF formulae, \(v\) be an assignment such that \(v(\phi_i) = \true\) for every \(i\), and \(\curly{p_i}\), $q$ be probabilities. If \(P \neq NP\), the minimal selection of variables for which updating their probabilities to \(q\) would reduce the maximal MES value cannot be done in polynomial time.
\end{proposition}

\begin{proof}
We will show that the problem is computationally hard in the case where the probabilities are all \(\frac{1}{2}\) and \(v \equiv \true\).

As in the previous proposition, we will show a reduction from the hitting set problem where the sets are equipotent. Let \(C = \range{S_1}{S_m}\) be a collection of equipotent subsets of a set \(S = \range{s_1}{s_n} \subset \mathbb{N}\) and let~\(l\) be a positive number. We need to determine if a hitting set of size at most~\(l\) exists.

For every~\(S_i\), define \(\phi_i = \bigdisji{s \in S_i}{}{s}\). As in the previous proposition, a selection of variables to improve would lower the maximal MES value iff it is a hitting set for the formulae, assuming that~$q$ is low enough: the MES value of each formula is the product of the probabilities of its variables; if a formula doesn’t contain a variable that is included in the selection, its MES value would remain the same; since the sets~$S_i$ are equipotent, it means that the maximal MES value would also remain the same. Therefore, again, if such a selection can be computed in PTIME, then a PTIME solution of the hitting set problem would result.
\end{proof}

\section{Experimental Study}
\label{sec:experiments}
We now describe the experimental study evaluating the usefulness of our indicators and the effectiveness of our algorithms.

\subsection{Implementation}\label{sec:implementation} We have implemented our framework in a prototype system (see~\cite{gitrepo} for the source code), using PostgreSQL~16 as the relational database along with ProvSQL~\cite{senellart2018provsql} for computing provenance expressions. Our algorithms were implemented in Python~3 using the NumPy~\cite{numpy} and Pandas~\cite{mckinney2010pandas} libraries. The CP-SAT solver from OR-Tools library~\cite{cpsatlp} was used for solving ILPs (in Alg.~\ref{alg:correct_tuple_mes}). See~\cite{schreiber2024cleaner} for a demonstration of an end-to-end interactive framework that incorporates our algorithms.

\paragraph*{MESReduce subroutines}
For $\mathsf{ImproveVerification}$ and the baseline algorithms, we use a majority-vote labeling process (as, e.g., in crowdsourcing). The cost of a label corresponds to the number of votes (e.g., crowd responses) required to achieve the desired error probability, and unless stated otherwise, the budget limit is 1000 votes. This enables consistent experimentation across different algorithms and parameter settings. We implement $\mathsf{FindImprovementSet}$ and $\mathsf{NextProbability}$ following Section~\ref{sec:uncertainty_red_algs_inst}, and: \begin{inparaenum}[(i)] \item we limit the search for risky variables to 10 seconds for efficiency; \item we solve the mentioned hitting set instance using greedy approximation; \item if multiple tuples obtain the maximal MES, we consider $\mu$ of them for selecting an improvement set -- a hyper-parameter defaulted to~$50$; see Section~\ref{sec:hyperparams_experiments}.\end{inparaenum}

\subsection{Experimental Setup}
\label{sec:experiments_setting}

\paragraph*{Databases}\phantomsection
\refstepcounter{subsubsection}
Our experiments were performed on NELL~\cite{mitchell2015nell} and TPC-H (https://www.tpc.org/tpch/, scale~1GB) databases. NELL contains real uncertain data gathered automatically from the internet; TPC-H contains synthetic data, and provides scale, diversity, and standardization for algorithmic evaluation of our system. For NELL, we used a workload of 8 queries created by~\cite{drien2023query} independently of this work. TPC-H includes a diverse workload with different types of queries; since our framework only supports SPJU queries, we ran the experiments with $8$~queries that could be naturally transformed to SPJU queries while preserving the semantics, e.g., by stripping outer aggregation and replacing GROUP BY with SELECT DISTINCT where possible; negation, in particular, is not supported in our framework. The query results and their provenance are highly diverse, with~5-300k output tuples and provenance size with 2-58506 variable occurrences in $k$-DNF per output tuple. Unless stated otherwise, we have randomly selected the output tuples of interest~$O\subseteq Q(D)$, by default, $\cardinality{O}=200$.

\paragraph*{Scenarios}
To evaluate our algorithms in different settings, three scenarios are examined:
\begin{enumerate}
    \item \textbf{Real-labels scenario (\texttt{RLBL})}: we use the real human labels provided for a subset of 1.5M tuples from NELL. If a tuple is labeled more than once, we treat the majority label as the ground truth. We use a state-of-the-art external tool to obtain the (naturally imperfect) initial verification.
    \item \textbf{Worst-case scenario (\texttt{WCS})}: we force a highly initial erroneous scenario by setting the ground-truth assignment to be all~\true{}, the initial error probabilities to be all~\(0.499\), and the initial verification assignment to be all~\false.\footnote{We have also examined the symmetric case with ground truth all~\false{} and assignment all~true{}; as this second scenario yielded less errors, we omit it.} This scenario also reflects settings in which no knowledge of the uncertainty of the initial verification steps is known.
    \item \textbf{Average-case scenario (\texttt{AVG})}: we randomly generate a ground truth correctness labeling; then, we choose initial error probabilities uniformly between \(0.2\) and \(0.499\); finally, we generate the initial verification assignment according to these probabilities.
\end{enumerate}

\paragraph*{Baselines}

Our MES metric and the derived risky-tuples indicator are defined over the output of a generic verification system. The \emph{MESReduce} algorithm builds on these indicators to complement existing verification tools by coordinating additional verification steps that reduce output uncertainty and provide provable worst-case guarantees, without making assumptions about the underlying data. Accordingly, our framework is not intended to replace existing systems but to guide and enhance their use, and is therefore not directly comparable to them. To the best of our knowledge, no existing algorithm addresses this setting. We therefore design a set of baseline strategies that iteratively select tuples for verification, against which we evaluate MESReduce. Unlike these baselines, MESReduce selects verification actions based on the MES and risky-tuple indicators in combination with provenance information.

Each baseline is parameterized by the target error probability for each step, \(p\). This introduces a trade-off: a higher probability lowers the verification cost per tuple and allows more tuples to be verified, while a lower one spends more budget per tuple to achieve higher certainty.

We categorize the baselines according to their use of information about the verified tuples; a baseline is \emph{Probability-Aware} if it uses error probabilities, \emph{Structure-Aware} if it uses provenance structure, and \emph{Agnostic} if it uses neither.
\begin{compactitem}
    \item \textbf{Random($p$)}: Selecting input tuples to verify uniformly at random -- \emph{Agnostic}.
    \item \textbf{Formula-Count Greedy($p$)}: Greedily selecting input tuples by the number of affected target output tuples in~$O\subseteq Q(D)$ -- \emph{Structure-Aware}.
    \item \textbf{Occurrences-Count Greedy($p$)}: Greedily selecting input tuples by the total number of occurrences of their annotating variables in the provenance expressions of tuples in~$O\subseteq Q(D)$ -- \emph{Structure-Aware}.
    \item \textbf{Probability Greedy($p$)}: Greedily selecting input tuples to verify by descending order of error probability -- \emph{Probability-Aware}.
\end{compactitem}

All algorithms are tested under the same budget limit.

\begin{figure}[t]
\centering
\includegraphics[width=0.5\linewidth, trim=0 2pt 0 0, clip]{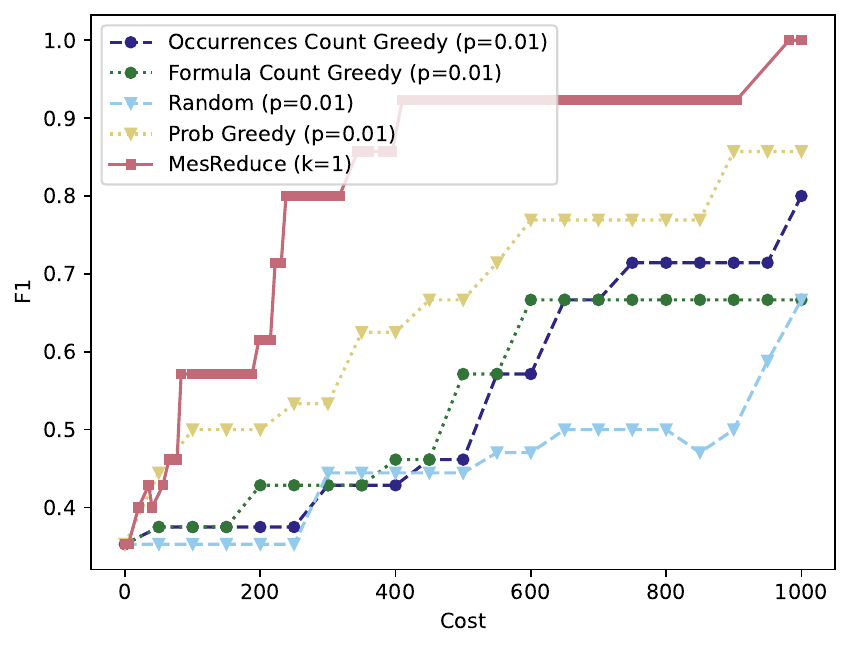}
\caption{\fscore{} metric as a function of the verification cost for selected algorithms on NELL's query Q4 in Scenario~\texttt{AVG}}
\label{fig:nell_q4_f1}
\end{figure}

\paragraph*{Metrics}
We evaluate both verification quality and scalability.
For quality, we use two complementary metrics:

\begin{compactitem}
\item\textbf{Maximal Error Score (MES):} Our worst-case uncertainty metric, which can increase or decrease during verification. We report the logarithmic ratio between the final and initial maximal MES values, where higher ratios indicate greater quality improvement, while the logarithmic form stabilizes computation.

\item\textbf{Extrinsic Accuracy (\fscore{}):} A measure comparing the computed correctness labels to the ground truth ones (unavailable to the algorithms themselves). Because the final-to-initial ratio of \fscore{} was less informative, we instead compute the area under the curve (AUC) of \fscore{} across verification steps.
\end{compactitem}

For illustration, Fig.~\ref{fig:nell_q4_f1} shows the \fscore{} increase during verification for NELL’s query Q4 in Scenario~\texttt{AVG}, where a sharper and earlier rise indicates better performance, as also reflected in the AUC. In this execution, MESReduce outperforms all other baselines in both AUC and final \fscore{} after~1000 steps. While Random and Formula-Count Greedy($p=0.01$) achieve similar final quality, the latter performs better across most budgets and consequently attains a higher AUC.

\subsection{End-to-End Workflows}
\label{sec:case_studies}

\begin{figure}
\fbox{\parbox{0.9\linewidth}{
\begin{smaller}
I have the following tuple from the NELL database: (entity="E", relation="R", value="V"). It represents a belief, which is a fact that can be correct or incorrect. Determine the correctness of this belief. For example, for the following facts you should answer "Correct": (entity="chocolate", relation="thinghascolor", value="brown"), (entity="detroit", relation="citylocatedinstate", value="michigan"), (entity="the\_andy\_warhol\_museum", relation="museumincity", value="pittsburgh"); for the following facts you should answer "Incorrect": (entity="arizona", relation="statehascapital", value="juneau"), (entity="mount\_katahdin", relation="mountaininstate", value="italy"), (entity="the\_da\_vinci\_code", relation="bookwriter", value="paul\_auster").
\end{smaller}
}}
\caption{LLM prompt for the verification of a tuple (entity="E", relation="R", value="V"). Concrete values of E, R, V are substituted.}
\label{fig:llm_verify_prompt}
\end{figure}

We conducted a few case studies to demonstrate end-to-end verification workflows using our framework (Fig.~\ref{fig:architecture_diagram}) in a real-world setting (Scenario~\texttt{RLBL}).

\paragraph*{Case Study (I): crowdsourced improvement of anomaly detection}
We begin with a case study in which initial verification is performed by an  anomaly detection tool, and MESReduce is used to decide \emph{which} input tuples to refine via crowdsourcing and \emph{how many} labels to collect per tuple.
An analyst executes query~Q3 from the workload of~\cite{drien2023query} over the NELL dataset (Fig.~\ref{fig:case_study_query}). The query retrieves teams, sports, and leagues linked through a common athlete and was selected due to the large number of relevant labels in the data. As an initial verification step, the analyst applies DTE~\cite{livernoche2024diffusion}, a recent, unsupervised anomaly detection model trained on structural properties of input tuples. (The extracted features per attribute: length, number of words, average word length, number of special characters, and number of unique characters.) The model is trained for 20 epochs, and the 10\% most anomalous tuples are labeled as incorrect.

\begin{figure}
\begin{smaller}
\begin{Verbatim}[frame=single]
SELECT DISTINCT a.value AS team,
                b.value AS sport,
                c.value AS league
FROM beliefs AS a, beliefs AS b, beliefs AS c
WHERE a.relation = 'concept:athleteplaysforteam'
  AND b.relation = 'concept:athleteplayssport'
  AND c.relation = 'concept:athleteplaysinleague'
  AND a.entity = b.entity
  AND a.entity = c.entity
\end{Verbatim}
\end{smaller}
\caption{NELL query \(Q_{3}\)}
\label{fig:case_study_query}
\end{figure}

Since DTE produces anomaly scores rather than calibrated error probabilities, we approximate error probabilities by mapping the anomaly threshold to a maximum error probability of~0.5 and linearly normalizing the remaining scores within the range $[0,0.5]$ based on their distance from the threshold. As shown below, even this coarse estimation is sufficient for our framework to operate effectively.

At this point, the analyst has preliminary verification results with unknown impact on the query output. Using our framework, she computes the MES to assess output reliability and identifies a highly uncertain output tuple, namely (\texttt{concept:sportsteam:rangers}, \texttt{concept:sport:baseball}, \texttt{concept:sportsleague:nhl}). Examining the risky-tuples indicator reveals that all input tuples contributing to this output are risky, indicating that refining any single input tuple in isolation would not reduce its uncertainty.

To systematically improve reliability under budget constraints, the analyst invokes the \emph{MESReduce} algorithm over the full set of query output tuples. MESReduce automatically selects input tuples for re-verification and determines how many crowdsourced labels to collect for each. Final labels are assigned by majority vote. As MESReduce progresses, the analyst observes a substantial reduction in the query MES log value, from $-2.079$ to $-8.316$, reflecting a significantly stronger worst-case reliability guarantee without making assumptions about the underlying data. This improvement is also reflected in extrinsic quality metrics (Fig.~\ref{fig:case_study_metrics}): precision increases from~$0.961$ to~$1.0$, recall from~$0.833$ to~$0.933$, and \fscore{} from~$0.893$ to~$0.965$.

\begin{figure}[t]
\center
\includegraphics[width=0.5\linewidth, trim=0 2pt 0 0, clip]{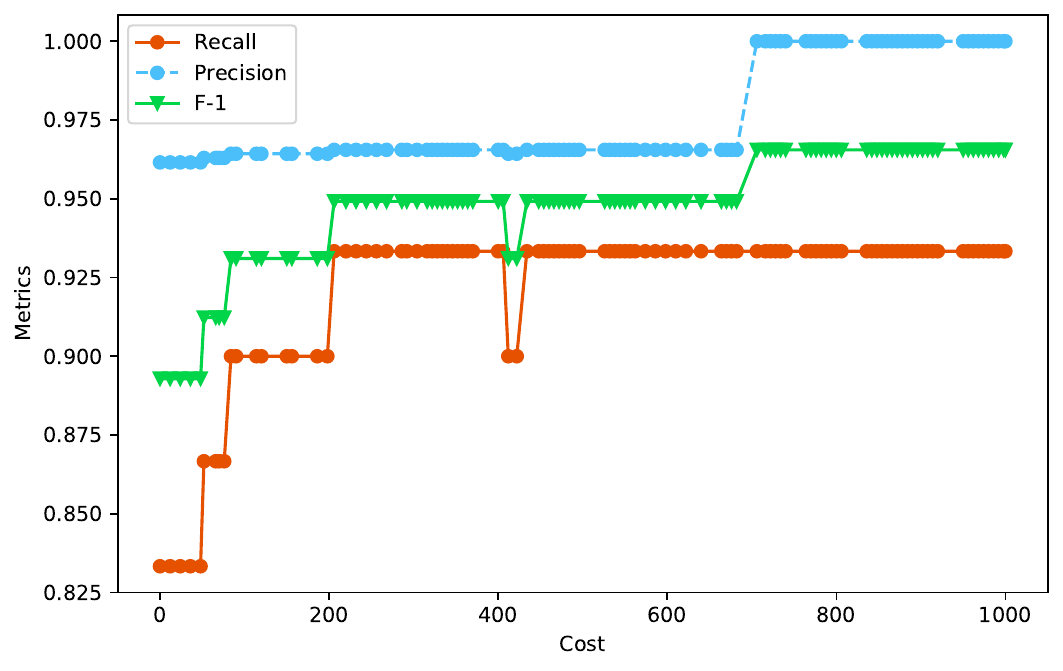}
\caption{The progress of the precision, recall, and \fscore\ scores in the case study experiment (Scenario~\texttt{RLBL})}
\label{fig:case_study_metrics}
\end{figure}

\paragraph*{Case Study~(II): Improving upon an LLM with crowdsourcing}
This case study illustrates how an analyst can use an LLM-based verifier to obtain initial verification results with error probabilities, and then apply MESReduce to guide targeted crowdsourcing for further improvement.

The analyst again executes NELL query~Q3, this time relying on a cost-effective LLM for initial verification. She selects Gemini~2.0 Flash-Lite, the smallest and most budget-friendly Gemini model available via API (\url{https://ai.google.dev/gemini-api/docs}). The analyst sends the relevant input tuples for verification in separate API calls, using a simple output schema of \emph{Correct} or \emph{Incorrect}. The prompt, shown in Figure~\ref{fig:llm_verify_prompt}, provides examples of correct and incorrect NELL tuples and asks the model to verify a given tuple.
To estimate error probabilities, the analyst uses the token-level probabilities of the two possible outputs~\cite{kadavath2022language, gemini_logprobs, openai_logprobs}, yielding values that range from $4.768\times10^{-6}$ to $0.487$. Using these estimates, she computes the MES log values for the query results and observes an initial overall MES of $-0.897$, indicating substantial uncertainty.

To improve reliability under budget constraints, the analyst invokes MESReduce on the output tuples, using crowdsourcing with a limit of~100 votes. MESReduce automatically selects which of the~322 relevant input tuples to refine and how much effort to allocate to each. After refinement, the maximal MES log value decreases to $-3.262$, indicating a substantially stronger worst-case reliability guarantee. Extrinsic quality metrics also improve: precision remains at $1.0$, recall increases from $0.9$ to $0.966$, and the \fscore{} rises from $0.947$ to $0.983$.

\paragraph*{Case Study~(III): Reducing the cost of automated verification}
Our third case study illustrates how MESReduce can be used to substantially reduce the cost of automated verification. Rather than applying an expensive, high-quality verifier to the entire database, which is costly and does not scale, the analyst first applies a fast, low-cost anomaly detector (DTE, as in Case Study~(I)) and then selectively invokes a high-quality LLM-based verifier using MESReduce. In this setting, the LLM is powered by the most capable Gemini model available via API, Gemini~3~Pro. MESReduce improves both the reliability and quality of the query results while requiring far fewer LLM invocations.

As in the first case study, the analyst applies DTE as an initial verification step for NELL query~Q3. Due to the limited accuracy of this automated model, the resulting verification exhibits relatively high uncertainty, reflected in elevated MES values. The analyst then invokes \emph{MESReduce} using the Gemini~3~Pro model, limiting the process to~100 API calls. In this scenario, the analyst does not have reliable error probability estimates for the LLM and instead treats its outputs as nearly perfect, assigning a uniform error probability of~0.

Despite this coarse and potentially inaccurate probability assumption, which may lead to suboptimal behavior compared to using precise probabilities, MESReduce still substantially reduces output uncertainty: the maximal MES log value decreases from $-2.085$ to $-7.611$. More importantly, extrinsic quality metrics, which are independent of the provided probabilities, also improve significantly: precision increases from $0.961$ to $1.0$, recall remains at $0.833$, and the F1 score improves from $0.892$ to $0.909$.

\subsection{Verification Quality}
\label{sec:verification_quality}
We now present experimental results that evaluate the quality improvement achieved by the different algorithms.

\paragraph*{MES reduction}
We examined the effect of the baseline algorithms on the MES in the average-case, Scenario~\texttt{AVG}. Table~\ref{tab:mes_reduction_experiment} presents the ratio of final to initial MES values in logarithmic form, over NELL and TPC-H, averaged over~\(100\) executions. As shown, MESReduce achieves the best improvement across all queries in both datasets. Its averaged ratios are all greater than~\(1\) (and~$\infty$ when the final MES is~\(0\)), indicating that MES is reduced and that query output reliability improves with a worst-case guarantee. In contrast, for some baselines and queries, the ratio is below~\(1\), meaning that MES increases and worst-case reliability degrades despite additional verification effort. The comparative performance of the baselines shows no clear trend.

\begin{table*}[t]
 \caption{Average ratio of final MES to initial MES (log value) in Scenario~\texttt{AVG}, the best-performing algorithm in bold.}
  \label{tab:mes_reduction_experiment}
 \begin{adjustbox}{max width=\linewidth}
  \begin{tabular}{*{10}{c}}
  \toprule
    \textbf{NELL Queries}& & \textbf{Q1} & \textbf{Q2} & \textbf{Q3} & \textbf{Q4} & \textbf{Q5} & \textbf{Q6} & \textbf{Q7} & \textbf{Q8} \\
    \midrule
    MESReduce & & \textbf{1.68} & \textbf{2.53} & \textbf{1.93} & \textbf{6.63} & \textbf{1.26} & \textbf{2.26} & \textbf{1.67} & \textbf{5.44} \\
    \midrule
    Random & \(p = 0.01\) & 0.74 & 0.84 & 0.84 & 0.63 & 0.89 & 0.88 & 0.62 & 0.56 \\
    \midrule
    Random & \(p = 0.0001\) & 0.74 & 0.92 & 0.91 & 0.69 & 0.95 & 0.94 & 0.70 & 0.70 \\
    \midrule
    Formula-Count & \(p = 0.01\) & 0.74 & 0.63 & 1.00 & 0.83 & 0.71 & 0.61 & 0.56 & 0.70 \\
    \midrule
    Formula-Count & \(p = 0.0001\) & 0.73 & 0.65 & 0.99 & 0.67 & 0.70 & 0.63 & 0.56 & 0.58 \\
    \midrule
    Occurrences-Count & \(p = 0.01\) & 0.74 & 0.83 & 1.00 & 0.84
& 0.71 & 0.71 & 0.56 & 0.64 \\
    \midrule
    Occurrences-Count & \(p = 0.0001\) & 0.73 & 0.88 & 1.00 & 0.66 & 0.70 & 0.78 & 0.56 & 0.63 \\
    \midrule
    Prob-Greedy & \(p = 0.01\) & 0.98 & 0.88 & 0.88 & 0.95 & 0.91 & 0.88 & 0.75 & 0.76 \\
    \midrule
    Prob-Greedy & \(p = 0.0001\) & 0.86 & 0.93 & 0.92 & 0.81 & 0.96 & 0.92 & 0.76 & 0.77 \\
    \bottomrule
  \end{tabular}

  \begin{tabular}{*{10}{c}}
    \toprule
    \textbf{TPC-H Queries}& & \textbf{Q3} & \textbf{Q4} & \textbf{Q5} & \textbf{Q6} & \textbf{Q7} & \textbf{Q8} & \textbf{Q9} & \textbf{Q10} \\
    \midrule
    MESReduce & & \textbf{1.59} & \(\infty\) & \(\infty\) & \textbf{2.93} & \(\infty\) & \(\infty\) & \textbf{1.33} & \textbf{1.41} \\ \midrule
    Random & \(p = 0.01\) & 0.74 & 1.00 & 1.00 & 1.00 & 0.82 & 0.92 & 0.83 & 0.85 \\ \midrule
    Random & \(p = 0.0001\) & 0.81 & 1.00 & 1.00 & 1.00 & 0.95 & 0.96 & 0.90 & 0.89 \\ \midrule
    Formula-Count & \(p = 0.01\) & 0.55 & 1.00 & 1.00 & 1.00 & 1.37 & 1.22 & 0.84 & 0.82 \\ \midrule
    Formula-Count & \(p = 0.0001\) & 0.59 & 1.00 & 1.00 & 1.00 & 2.44 & 1.86 & 0.84 & 0.80 \\ \midrule
    Occurrences-Count & \(p = 0.01\) & 0.56 & 1.00 & 1.00 & 1.00 & 1.34 & 1.24 & 0.84 & 0.81 \\ \midrule
    Occurrences-Count & \(p = 0.0001\) & 0.62 & 1.00 & 1.00 & 1.00 & 2.44 & 1.87 & 0.85 & 0.81 \\ \midrule
    Prob-Greedy & \(p = 0.01\) & 0.80 & 1.00 & 1.00 & 1.48 & 0.79 & 0.96 & 0.84 & 0.85 \\ \midrule
    Prob-Greedy & \(p = 0.0001\) & 0.86 & 1.00 & 1.00 & 1.22 & 0.92 & 1.04 & 0.90 & 0.89 \\
    \bottomrule
  \end{tabular}
  \end{adjustbox}
\end{table*}

\begin{table*}[t]
 \caption{Worst \fscore{} AUC scores in Scenario~\texttt{WCS}}
  \label{tab:worst_auc_scenario_1_experiment}
 \begin{adjustbox}{max width=\linewidth}
  \begin{tabular}{*{10}{c}}
    \toprule
  \textbf{NELL Queries} & & \textbf{Q1} & \textbf{Q2} & \textbf{Q3} & \textbf{Q4} & \textbf{Q5} & \textbf{Q6} & \textbf{Q7} & \textbf{Q8} \\
    \midrule
    MESReduce & & \textbf{342.63} & \textbf{201.46} &  \textbf{148.70} & \textbf{587.66} & \textbf{61.03} &  \textbf{101.10} & \textbf{181.57} & \textbf{438.82} \\
    \midrule
    Random & \(p = 0.01\) & 34.09 & 0.00 & 0.00 & 105.28 & 0.00 & 0.00 & 0.00 & 15.34 \\
    \midrule
    Random & \(p = 0.0001\) & 2.23 & 0.00 & 0.00 & 0.00 & 0.00 & 0.00 & 0.00 & 0.00 \\
    \midrule
    Formula-Count & \(p = 0.01\) & 74.03 & 0.00 & 8.12 & 387.08 & 7.47 & 0.00 & 0.00 & 301.64 \\
    \midrule
    Formula-Count & \(p = 0.0001\) & 0.00 & 0.00 & 3.73 & 188.67 & 0.00 & 0.00 & 0.00 & 91.08 \\
    \midrule
    Occurrences-Count & \(p = 0.01\) & 72.29 & 0.00 & 12.66 & 351.04 & 12.01 & 0.00 & 0.00 & 204.61 \\
    \midrule
    Occurrences-Count & \(p = 0.0001\) & 0.00 & 0.00 & 12.37 & 148.35 & 11.14 & 0.00 & 0.00 & 83.45 \\
    \midrule
    Prob-Greedy & \(p = 0.01\) & 38.43 & 0.00 & 0.00 & 86.85 & 0.00 & 0.00 & 0.00 & 16.51 \\
    \midrule
    Prob-Greedy & \(p = 0.0001\) & 4.76 & 0.00 & 0.00 & 0.00 & 0.00 & 0.00 & 0.00 & 0.00 \\
    \bottomrule
  \end{tabular}

  \begin{tabular}{*{8}{c}}
    \toprule
   \textbf{TPC-H Queries} & & \textbf{Q3} & \textbf{Q6} & \textbf{Q7} & \textbf{Q8} & \textbf{Q9} & \textbf{Q10} \\
    \midrule
    MESReduce & & \textbf{189.19} & \textbf{599.71} & \textbf{28.02} & 8.06 & 2.45 & \textbf{68.01} \\
    \midrule
    Random & \(p = 0.01\) & 0.00 & 354.86 & 0.00 & 0.00 & 0.00 & 0.00 \\
    \midrule
    Random & \(p = 0.0001\) & 0.00 & 212.5 & 0.00 & 0.00 & 0.00 & 0.00 \\
    \midrule
    Formula-Count & \(p = 0.01\) & 0.25 & 339.26 & 0.00 & 0.00 & 10.78 & 26.03 \\
    \midrule
    Formula-Count & \(p = 0.0001\) & 0.00 & 193.35 & 0.00 & 0.00 & 0.00 & 3.82 \\
    \midrule
    Occurrences-Count & \(p = 0.01\) & 0.00 & 337.94 & 0.00 & \textbf{11.97} & \textbf{11.97} & 0.00 \\
    \midrule
    Occurrences-Count & \(p = 0.0001\) & 0.00 & 193.35 & 0.00 & 0.00 & 0.00 & 0.00 \\
    \midrule
    Prob-Greedy & \(p = 0.01\) & 0.00 & 355.98 & 0.00 & 0.00 & 0.00 & 0.00 \\
    \midrule
    Prob-Greedy & \(p = 0.0001\) & 0.00 & 213.43 & 0.00 & 0.00 & 0.00 & 0.00 \\
    \bottomrule
  \end{tabular}
\end{adjustbox}
\end{table*}

\paragraph*{Extrinsic Accuracy}
We assessed the impact on label accuracy as measured by the AUC of the \fscore{} score, comparing to the ground truth. Focusing on worst-case guarantees, we first examined Scenario~\texttt{WCS} and the worst AUC obtained across~100 executions of each algorithm and query. The results for NELL and TPC-H are shown in Table~\ref{tab:worst_auc_scenario_1_experiment}.\footnote{TPC-H queries Q4 and Q5, with few output tuples, were omitted from this experiment as their AUC was uninformative.} Except for TPC-H queries Q8 and Q9, MESReduce consistently achieves the best score, demonstrating that reducing the MES effectively improves the extrinsic accuracy in the worst-case scenario.

\begin{table*}[t]
\caption{Worst \fscore{} AUC scores in Scenario~\texttt{AVG}}  \label{tab:worst_auc_scenario_2_experiment}
 \begin{adjustbox}{max width=\linewidth}
  \begin{tabular}{*{10}{c}}
    \toprule
\textbf{NELL Queries} & & \textbf{Q1} & \textbf{Q2} & \textbf{Q3} & \textbf{Q4} & \textbf{Q5} & \textbf{Q6} & \textbf{Q7} & \textbf{Q8} \\
    \midrule
    MESReduce & & \textbf{531.88} & 377.22 & \textbf{533.93} & \textbf{549.50} & 219.45 & 236.72 & \textbf{383.03} & \textbf{431.67} \\
    \midrule
    Random & \(p = 0.01\) & 335.88 & 155.25 & 452.95 & 67.73 & 0.00 & 0.00 & 76.62 & 0.00 \\
    \midrule
    Random & \(p = 0.0001\) & 281.29 & 172.19 & 458.31 & 43.18 & 0.00 & 0.00 & 45.00 & 0.00 \\
    \midrule
    Formula-Count & \(p = 0.01\) & 419.46 & 208.67 & 465.73 & 263.17 & \textbf{246.79} & 86.69 & 294.52 & 0.00 \\
    \midrule
    Formula-Count & \(p = 0.0001\) & 392.39 & 159.67 & 456.85 & 182.80 & 242.44 & 67.47 & 207.44 & 0.00 \\
    \midrule
    Occurrences-Count & \(p = 0.01\) & 418.60 & \textbf{386.28} & 466.55 & 272.43 & 211.15 & \textbf{249.17} & 308.13 & 0.00 \\
    \midrule
    Occurrences-Count & \(p = 0.0001\) & 392.39 & 259.58 & 458.40 & 185.15 & 242.44 & 101.22 & 206.82 & 0.00 \\
    \midrule
    Prob-Greedy & \(p = 0.01\) & 379.11 & 101.44 & 423.33 & 0.00 & 0.00 & 0.00 & 99.85 & 0.00 \\
    \midrule
    Prob-Greedy & \(p = 0.0001\) & 314.17 & 89.12 & 416.10 & 0.00 & 0.00 & 0.00 & 54.90 & 0.00 \\
    \bottomrule
  \end{tabular}

  \begin{tabular}{*{8}{c}}
    \toprule
\textbf{TPC-H Queries} & & \textbf{Q3} & \textbf{Q6} & \textbf{Q7} & \textbf{Q8} & \textbf{Q9} & \textbf{Q10} \\
    \midrule
    MESReduce & & 411.68 & \textbf{733.60} & \textbf{145.85} & 0.00 & 0.00 & \textbf{219.78} \\
    \midrule
    Random & \(p = 0.01\) & 143.86 & 664.35 & 0.00 & 0.00 & 0.00 & 0.00 \\
    \midrule
    Random & \(p = 0.0001\) & 131.66 & 632.65 & 0.00 & 0.00 & 0.00 & 0.00 \\
    \midrule
    Formula-Count & \(p = 0.01\) & 432.61 & 659.94 & 0.00 & 0.00 & 0.00 & 138.15 \\
    \midrule
    Formula-Count & \(p = 0.0001\) & 314.07 & 621.11 & 0.00 & 0.00 & 0.00 & 112.49 \\
    \midrule
    Occurrences-Count & \(p = 0.01\) & \textbf{451.69} & 659.94 & 0.00 & 0.00 & 0.00 & 215.77 \\
    \midrule
    Occurrences-Count & \(p = 0.0001\) & 323.45 & 621.11 & 0.00 & 0.00 & 0.00 & 168.04 \\
    \midrule
    Prob-Greedy & \(p = 0.01\) & 165.43 & 710.45 & 0.00 & 0.00 & 0.00 & 53.10 \\
    \midrule
    Prob-Greedy & \(p = 0.0001\) & 151.22 & 650.67 & 0.00 & 0.00 & 0.00 & 18.79 \\
    \bottomrule
  \end{tabular}
\end{adjustbox}
\end{table*}

The worst results for Scenario~\texttt{AVG} are shown in Table~\ref{tab:worst_auc_scenario_2_experiment} for NELL and TPC-H. In this scenario, with fewer output label errors, error detection and correction are more challenging. Nevertheless, MESReduce remains robust, outperforming in most queries, and achieving comparable performance in others. For this quality metric, baselines with a higher $p$ parameter usually achieved larger improvements.

In summary, these experiments demonstrate that, despite the MES focusing on a worst-case scenario rather than the most likely world, reducing the MES effectively enhances accuracy and reduces errors in data verification.

\subsection{Scalability Experiments}
\label{sec:scalability_experiments}

\begin{table*}[t]
\caption{Scalability experiments: average MES computation time (milliseconds)}
\label{tab:scalability_experiments}
 \begin{adjustbox}{max width=\linewidth}
  \begin{tabular}{*{9}{c}}
    \toprule
 \textbf{NELL Queries} & \textbf{Q1} & \textbf{Q2} & \textbf{Q3} & \textbf{Q4} & \textbf{Q5} & \textbf{Q6} & \textbf{Q7} & \textbf{Q8} \\
    \midrule
    Correct tuple & 2.094 & 2.605 & 2.321 & 2.170 & 2.693 & \textbf{2.751} & 2.114 & 2.324 \\
    Incorrect tuple & 0.016 & 0.073 & 0.042 & 0.024 & 0.077 & \textbf{0.093} & 0.017 & 0.034 \\
    \bottomrule
  \end{tabular}
  \begin{tabular}{*{9}{c}}
    \toprule
 \textbf{TPC-H Queries}& \textbf{Q3} & \textbf{Q4} & \textbf{Q5} & \textbf{Q6} & \textbf{Q7} & \textbf{Q8} & \textbf{Q9} & \textbf{Q10} \\
    \midrule
    Correct tuple & 2.242 & \textbf{1450.629} & 466.495 & 2.017 & 2.160 & 2.216 & 2.285 & 2.456 \\
    Incorrect tuple & 0.029 & \textbf{152.093} & 68.364 & 0.013 & 0.024 & 0.028 & 0.024 & 0.037 \\
    \bottomrule
  \end{tabular}
  \end{adjustbox}
\end{table*}

\begin{figure}[t]
\center
\subfloat[Running Time Experiment]{%
  \includegraphics[width=.4\linewidth]{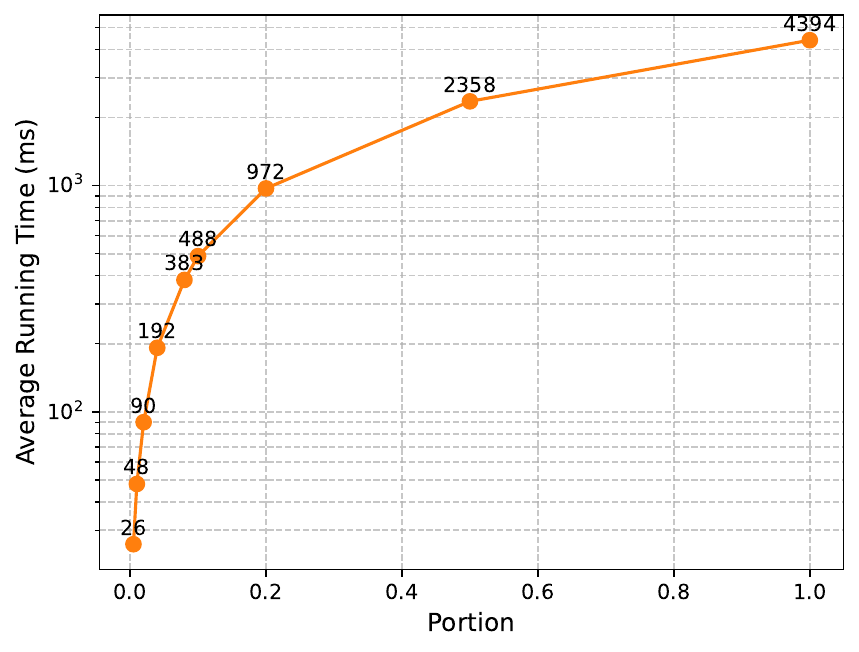}%
}
\hfil
\subfloat[Memory Usage Experiments]{%
  \includegraphics[width=.4\linewidth]{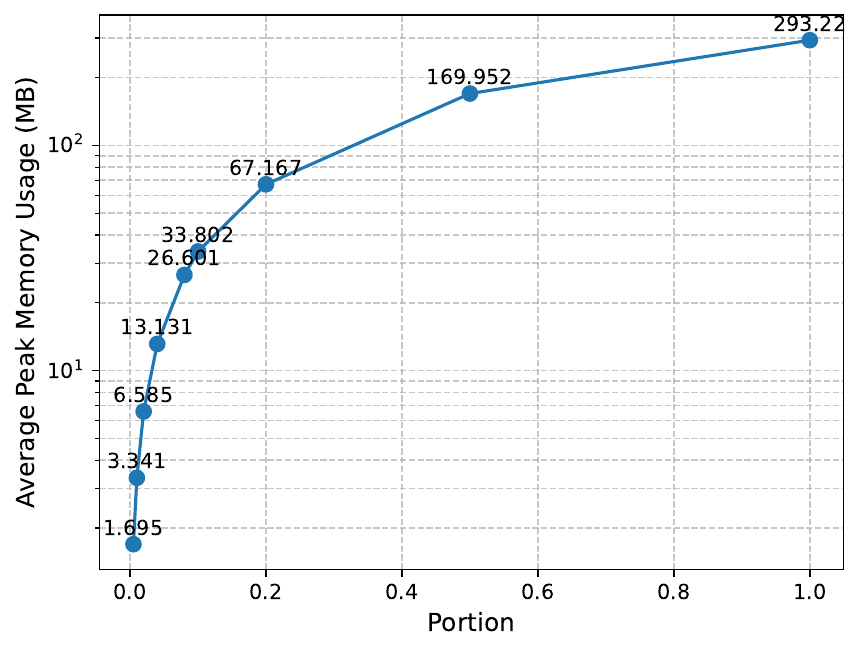}%
}
\caption{Scalability Experiment}
\label{fig:scalability_exp}
\end{figure}

\begin{figure*}[t]
\center
\subfloat{%
  \includegraphics[width=.4\textwidth]{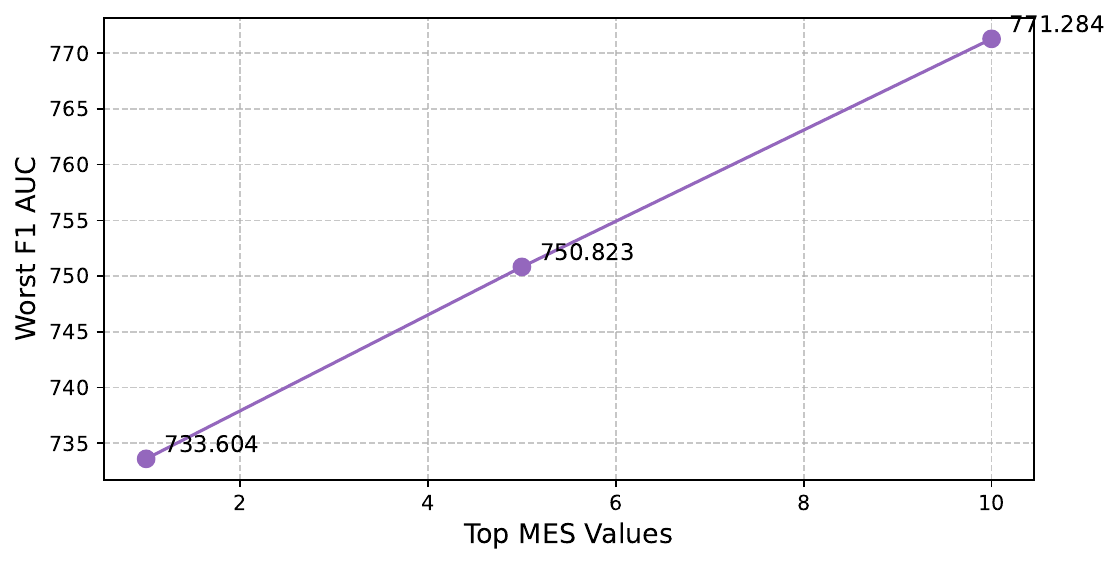}%
}
\hfil
\subfloat{%
  \includegraphics[width=.4\textwidth]{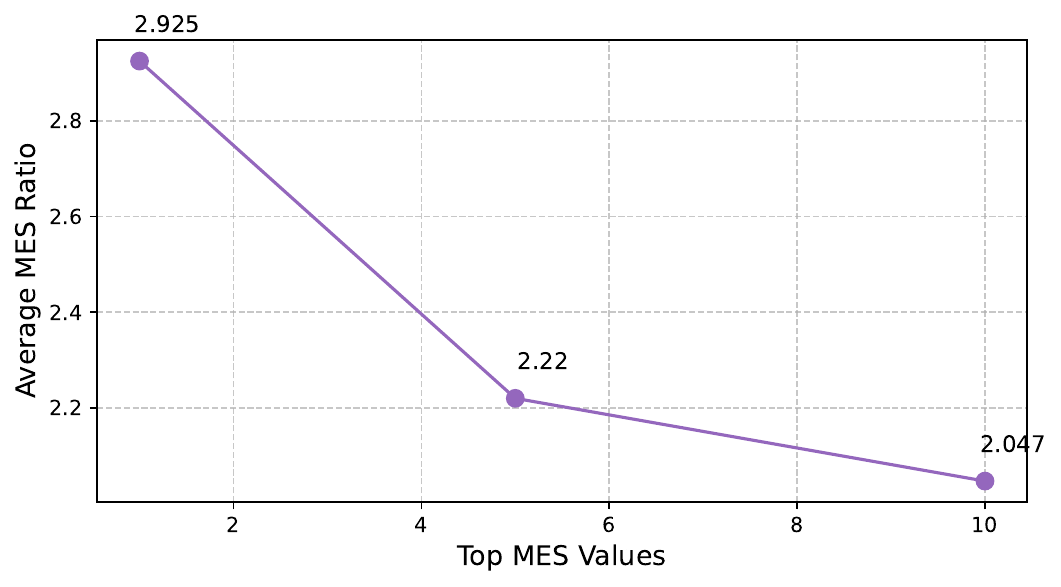}%
}
\caption{Top-k MES values Experiment (TPC-H Q6, Scenario~\texttt{AVG})}
\label{fig:top_k_exp}
\end{figure*}

\begin{figure*}[t]
\center
\subfloat{%
  \includegraphics[width=.4\linewidth]{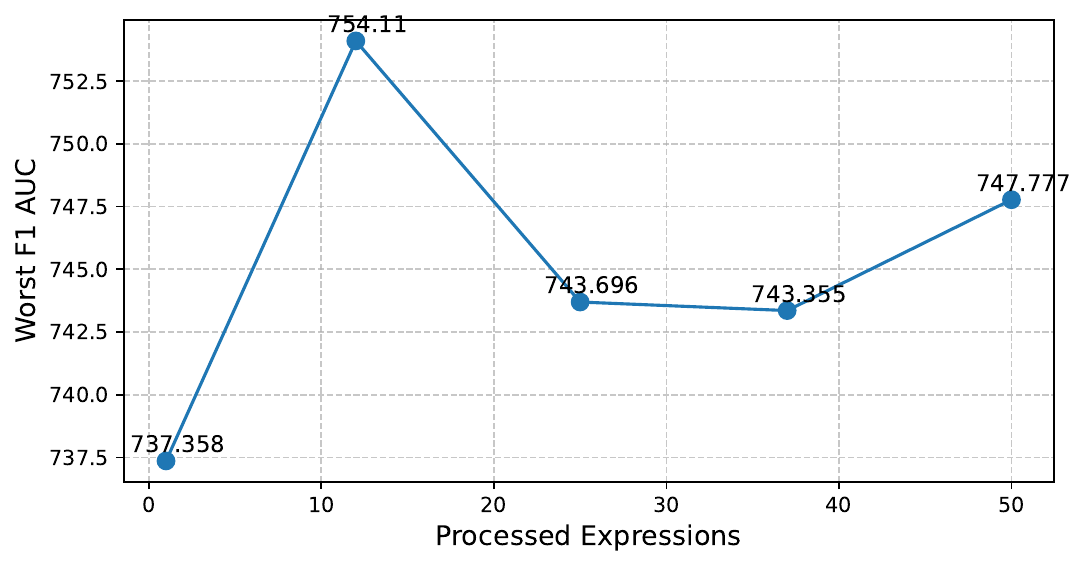}%
}
\hfil
\subfloat{%
  \includegraphics[width=.4\linewidth]{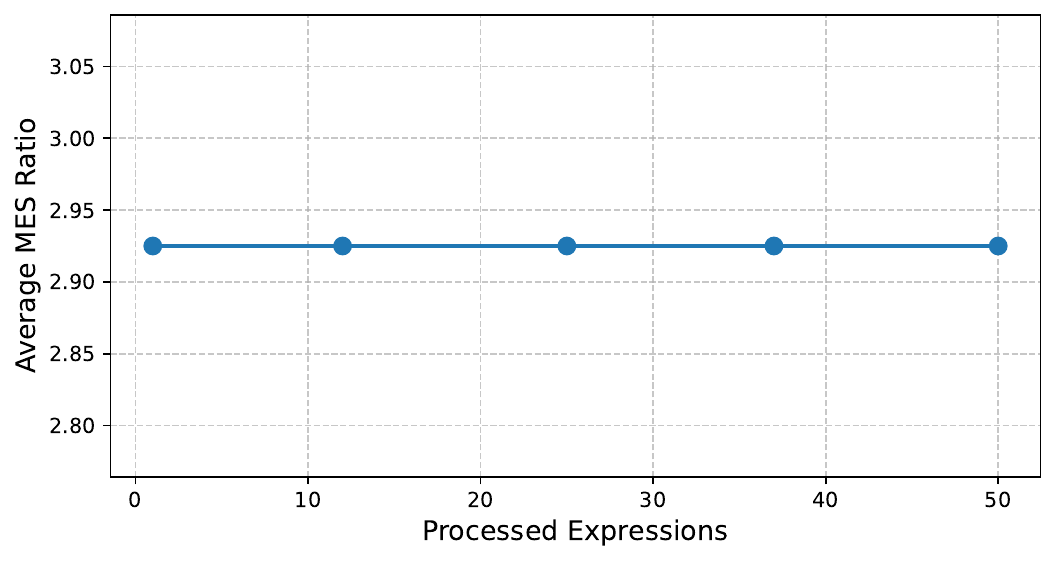}%
}
\caption{Tie-Breaking Experiment (TPC-H Q6, Scenario~\texttt{AVG})}
\label{fig:tie_breaking_exp}
\end{figure*}

\paragraph*{MES Computation}
These experiments evaluate the scalability of our MES computation algorithms (Section~\ref{sec:computation}). Since different algorithms are used depending on the output label, we randomly generated satisfying and non-satisfying assignments for each query, yielding~\true{} and~\false{} labels, respectively, and measured the MES computation time for every output tuple. We then computed the mean runtime separately for correct tuples (ILP solver from Algorithm~\ref{alg:correct_tuple_mes}) and incorrect ones (PTIME Algorithm~\ref{alg:incorrect_tuple_mes}). The procedure was repeated ten times, and the averaged results are reported in Table~\ref{tab:scalability_experiments}. All measurements are in milliseconds, with the maximal value in each row highlighted in bold.

For most queries, the computation is highly efficient, taking less than 3~milliseconds. Only TPC-H queries~Q4 and~Q5 required substantially longer running times due to the longer provenance expressions per tuple; for Q4, the average computation exceeded one second per correct tuple, which, nevertheless, remains reasonable. Unsurprisingly, the PTIME computation for output tuples labeled as incorrect is consistently faster than the ILP solver execution for correct ones.

\paragraph*{ILP Solver}
To further assess our usage of the ILP solver, we constructed a stress test by concatenating, via disjunction, three provenance expressions from the most computation-intensive query, TPC-H~Q4. The resulting expression consisted of \(118644\) variables and \(87084\) terms. We randomly generated a satisfying assignment and corresponding probabilities, then iteratively increased the fraction of terms considered ($0.005$, $0.01$, $0.02$, $0.04$, $0.08$, $0.1$, $0.2$, $0.5$, $1$). For each portion, we averaged ten MES computations. Fig.~\ref{fig:scalability_exp} presents the average running time and memory consumption. The runtime ranged from~25~milliseconds to~4.3~seconds, and memory usage from~1.7~MB to~293~MB, demonstrating that the algorithm scales well.
\subsection{Output Tuples Hyperparameters}
\label{sec:hyperparams_experiments}

MESReduce has two generalizations that consider more output tuples per iteration: one selects the top-$k$ tuples by MES, and the other includes up to $\mu$ tuples tied for the highest value. We fixed an arbitrary query, TPC-H~Q6, and ran MESReduce in Scenario~\texttt{AVG}, varying one hyperparameter at a time (\(k \in \curly{1,5,10}\), \(\mu \in \curly{1, 12, 25, 37, 50}\)). The aggregated results of $100$~independent runs are depicted in Figures~\ref{fig:top_k_exp},~\ref{fig:tie_breaking_exp} for the first and second generalizations, resp. Both variants yield larger improvement sets, resulting in fewer but more computationally intensive iterations and reduced adaptability during verification. They slightly improve extrinsic accuracy but show minor degradation in MES reduction; overall quality remains comparable, with no clear trend. Hence, these variants could be suitable for high-latency yet highly parallelizable verifiers (e.g., crowdsourcing). We plan to further study this trade-off between iteration count and adaptability in future work.

\section{Related work}
\label{sec:related_work}
\paragraph*{Data Quality}
Our work relates to data verification, anomaly detection, and data cleaning. The external verifier used by \emph{MESReduce} may come from any of these domains.

\emph{Data verification} assesses tuple correctness, typically via expert labeling by humans (e.g.,~\cite{mahdavi2019raha,nguyen2019user}) or trained machine learning models (e.g.,~\cite{dong2025automated,mahdavi2019raha,redyuk2021automating}), or indirectly through data ``unit tests'' (e.g.,~\cite{breck2019data,bylois2024data,caveness2020tensorflow,chen2025towards,schelter2018automating}). In some domains, correctness is assessed by cross-checking overlapping data sources (e.g.,~\cite{tian2019automated}). 
\emph{Anomaly detection} instead flags potentially incorrect tuples based on statistical or structural irregularities (e.g.,~\cite{chang2023data,chen2025towards,dong2025automated,gopalan2019pidforest,livernoche2024diffusion,liu2008isolation,mahdavi2019raha,pang2019deep,shankar2023automatic,redyuk2021automating}). Finally, large language models (LLMs) can be used as verifiers`\cite{chen2025towards}, as demonstrated in Section~\ref{sec:case_studies}. Our framework supports two complementary uses of LLM-based verification. First, it enables uncertainty-aware analysis of LLM-provided labels, accounting for factors such as missing or outdated training data. Even in our case studies, which involve largely general-domain data, we further improve the reliability and accuracy of LLM verification labels through targeted crowdsourcing (Case Study~(II)). Second, the framework helps reduce the cost of LLM-based verification by combining a fast, low-cost automated verifier with selective LLM calls guided by \emph{MESReduce} (Case Study~(III)).

Data verification is a component of the broader data cleaning problem, which addresses incorrect or inconsistent data (see~\cite{ilyas2019data}). Prior work typically defines clean data via constraints and focuses on identifying and repairing violations (e.g.,~\cite{afrati2009repair,bergman2015query,chu2013holistic,desa2019formal,geerts2020cleaning,livshits2018computing,rezig2021horizon}), or assumes conformance to statistical distributions (e.g.,~\cite{yakout2013scared}). Techniques include optimal repair computation (e.g.,~\cite{afrati2009repair,bertossi2013data,chu2013holistic,desa2019formal,geerts2020cleaning,gilad2020multiple,livshits2018computing,rezig2021horizon}), machine learning–based approaches (e.g.,~\cite{neutatz2021cleaning,rekatsinas2017holoclean}), and human-in-the-loop methods (e.g.,~\cite{assadi2018cleaning,bergman2015query,cheng2008cleaning,cheng2013cleaning,chu2015katara,drien2023query,fariha2021coco,krishnan2016activeclean,lin2018human,wang2014sample,yakout2011guided}).

Among existing systems, two widely cited approaches are HoloClean~\cite{rekatsinas2017holoclean} and ActiveClean~\cite{krishnan2016activeclean}. Both systems focus on \emph{data repair}, whereas our framework targets \emph{data verification} and operates on top of existing verifiers rather than replacing them. HoloClean combines constraints, external sources, and statistical models to repair erroneous data, while ActiveClean integrates cleaning into convex statistical learning workflows (e.g., SVMs) via stochastic gradient descent. Because our framework analyzes how verification errors propagate to query results and provides worst-case reliability guarantees, without performing repairs itself, a direct performance comparison with these end-to-end cleaning systems is not meaningful. Nevertheless, it may be interesting future work to adapt the MES metric to data repair or learning-based cleaning settings, including machine learning pipelines.

Closely related is \emph{query-guided data cleaning}, which focuses on cleaning query-relevant subsets rather than entire datasets (e.g.,~\cite{bergman2015query,cheng2008cleaning,drien2023query,mo2013cleaning,wang2014sample}). Rather than proposing a new verification or cleaning method, this work analyzes how verification errors propagate to query results and introduces algorithms that provide query-relevant reliability insights for existing verification systems.

\paragraph*{Uncertain and Probabilistic Databases}
A large body of work studies uncertainty in databases 
(e.g.,~\cite{arenas1999consistent,arenas2003answer,benjelloun2008databases,bertossi2011database,dalvi2007efficient,desa2019formal,greco2017computing,lian2010consistent,papaioannou2018supporting,suciu2011prob,suciu2020probabilistic,van2017query,wijsen2019foundations}). \emph{Probabilistic databases} model uncertainty by treating tuples or attributes as random variables.
For example, tuple-independent probabilistic databases assign each tuple an independent probability of being present in a possible world. A formal model of unclean databases that incorporates both probabilistic data and probabilistic cleaning errors was proposed in~\cite{desa2019formal}.

Our formal model differs fundamentally from this line of work. Probabilistic models typically assume knowledge of the ground truth data distribution and probabilities, which often requires domain expertise and ongoing maintenance. In contrast, our approach assumes knowledge only of the verification steps applied to the data, making it substantially more practical. Moreover, while prior work focuses primarily on expectation-based or distribution-based guarantees, our objective is to provide strong worst-case guarantees suitable for data-critical settings. Finally, we explicitly study how gradually reducing verification error probabilities affects query results, whereas most prior work assumes these probabilities are fixed.

Some approaches to uncertain data also adopt a worst-case perspective, most notably \emph{certain answers} and \emph{consistent query answering}, which compute query results that hold in every possible world or every minimal repair of an inconsistent database~\cite{arenas1999consistent,arenas2003answer,gilad2020multiple,lian2010consistent,wijsen2019foundations,bertossi2011database}. These works differ from ours in that consistency is defined via integrity constraints on clean data, leading to fundamentally different analyses and solutions.

\paragraph*{Possibilistic Query Answering}
An alternative approach to handling uncertainty is based on possibility theory and fuzzy sets, with \emph{fuzzy databases} as the central formal model. Several variants have been explored, including fuzzy relations and possibility distributions over attribute values (see, e.g.,~\cite{bosc1997introduction,dubois1988possibility}). In contrast, our approach does not model data uncertainty at all; instead, it focuses exclusively on uncertainty arising from verification processes and remains agnostic to the underlying data values. Furthermore, our uncertainty metric MES is not only used to quantify uncertainty but also to actively guide verification decisions, as described in Section~\ref{sec:uncertainty_red_algs}.

\paragraph*{Provenance in databases}
Provenance, also called lineage, plays a central part in this research, and all our algorithms process provenance expressions that are represented by Boolean formulae. In general, provenance refers to additional information attached to output tuples about their derivation from the input tuples~\cite{cheney2009provenance, senellart2019provenance}. There are multiple definitions for provenance, and in this work, we are using an extension of the PosBool semiring of~\cite{green2007provenance,imielinski1984incomplete} to three-valued logic. In general, many works in the field of probabilistic and uncertain databases use Provenance. In addition, it has been used in various contexts, including data cleaning~\cite{bergman2015query,drien2023query,glavic2021data,parulian2022dcm}, consent management~\cite{drien2021managing}, access control~\cite{moffitt2015collaborative}, causality~\cite{salimi2016quantifying, meliou2010the}, and tuples contribution quantification~\cite{deutch2022computing, abramovich2024banzhaf}. Many types of provenance have been implemented in ProvSQL~\cite{senellart2018provsql}, which has been used in our experiments.

\section{Conclusions and future work}
\label{sec:conclusions}

In this paper, we studied verification in data-critical settings and introduced \emph{MES}, a novel worst-case metric for the uncertainty of query results. We analyzed its computational complexity and developed efficient algorithms for its calculation. We also identified and formalized the notion of \emph{risky tuples} -- input tuples for which further verification may increase output uncertainty -- and proposed an algorithm to detect them. Building on these insights, we presented a generic algorithm, \emph{MESReduce}, that operates on top of existing verification systems to improve the reliability of query results. We demonstrated the practicality and effectiveness of our approach through an extensive experimental study.

This work opens several promising directions for future research. First, we have preliminary results on alternative uncertainty metrics over sets of tuples (e.g., worst-case \emph{expected} number of errors). Second, our approach could be extended to broader query types beyond SPJU, e.g., aggregate queries, as well as to generalized notions of risky tuples. Moreover, while the verified results of SPJU queries can be used as training data for machine learning models, it would be interesting to further explore the application of our techniques and algorithms within machine learning workflows -- specifically, the potential uses of MES, the relation between risky tuples and noise in training sets, and the impact of lowering the MES on model reliability. In addition, supporting additional classes of verifiers, such as constraint-based ones, could also broaden applicability. Finally, an interesting direction is to model the verifier's uncertainty using possibility theory, and apply our techniques and ideas under this alternative uncertainty model.

\paragraph*{Acknowledgments} ChatGPT was used for minor textual editing of some of the paragraphs in this submission, without affecting or generating content, non-textual parts (formulas, figures, tables...), or any of our contributions.

\bibliographystyle{IEEEtranS}
\bibliography{bib}

\end{document}